%% file: main.tex
\documentclass[11pt]{article}
\usepackage[top=1in, bottom=1in, left=1in, right=1in]{geometry}

\usepackage{booktabs} % For formal tables
\usepackage[ruled]{algorithm2e} % For algorithms
\usepackage{url}

\SetAlFnt{\small}
\SetAlCapFnt{\small}
\SetAlCapNameFnt{\small}
\SetAlCapHSkip{0pt}
\IncMargin{-\parindent}

\usepackage{xcolor}
\usepackage[colorlinks=true, allcolors=teal]{hyperref}
\usepackage{amsmath,amsthm,amsfonts}
\usepackage{mathtools}
\usepackage{cleveref}
\usepackage{nicefrac}
\usepackage{forest}
\usepackage{amssymb}

\DeclareMathOperator{\plu}{plu}
\newcommand{\bsigma}{{\boldsymbol{\sigma}}}
\newcommand{\bsigmaSTV}{\bsigma_{STV}}
\newcommand{\balpha}{{\boldsymbol{\alpha}}}
\newcommand{\profiledist}{\Pi}
\DeclarePairedDelimiter{\set}{\{}{\}}
\newcommand{\spn}{\textup{span}}
\DeclareMathOperator{\score}{sc}
\newcommand{\atob}{{a \leftrightarrow b}}

\DeclareMathOperator{\E}{\mathbb{E}}
\newcommand{\unif}{\text{unif}}
\newcommand{\Unif}{\text{Unif}}
\newcommand{\pos}{\txt{pos}}

\newtheorem{lemma}{Lemma}

\newtheorem{theorem}{Theorem}

\DeclarePairedDelimiter{\floor}{\lfloor}{\rfloor}

\usepackage{xcolor}
\definecolor{green}{HTML}{009900}

\input{JTex6}

\usepackage{natbib}
\allowdisplaybreaks

\setcitestyle{authoryear}

\begin{document}
\title{Computing Voting Rules with Elicited Incomplete Votes}
\author{        
        Daniel Halpern \\ 
        Harvard University \\ 
        \texttt{dhalpern@g.harvard.edu}
        \and
        Safwan Hossain \\ 
        Harvard University \\ 
        \texttt{shossain@g.harvard.edu}\vspace{0.7em}
        \and 
        Jamie Tucker-Foltz \\ 
        Harvard University\\ 
        \texttt{jtuckerfoltz@gmail.com} 
}

\date{}

\maketitle

\begin{abstract}
    Motivated by the difficulty of specifying complete ordinal preferences over a large set of $m$ candidates, we study voting rules that are computable by querying voters about $t < m$ candidates. Generalizing prior works that focused on specific instances of this problem, our paper fully characterizes the set of positional scoring rules that can be computed for any $1 \leq t < m$, which, notably, does not include plurality. We then extend this to show a similar impossibility result for single transferable vote (elimination voting). These negative results are information-theoretic and agnostic to the number of queries. Finally, for scoring rules that are computable with limited-sized queries, we give parameterized upper and lower bounds on the number of such queries a deterministic or randomized algorithm must make to determine the score-maximizing candidate. While there is no gap between our bounds for deterministic algorithms, identifying the exact query complexity for randomized algorithms is a challenging open problem, of which we solve one special case.
\end{abstract}

\maketitle

\section{Introduction}\label{secIntro}

Traditional social choice frameworks typically assume that voting rules have access to each voter's complete ordinal preferences over all candidates. Indeed, this is seen in practice as well with the widening adoption of ranked-choice voting systems~\citep{FairVote}, requiring voters to submit such information. Whether such information can actually be reliably elicited depends significantly on the context. If the number of candidates is small and voters have strong opinions, it may indeed be reasonable for them to provide a complete ranking. However, these assumptions do not hold in many scenarios. Primary elections in the United States, for example, routinely field large numbers of candidates, with many being unfamiliar to voters \citep{hirano2019primary}. 

A classic line of work in behavioral economics and psychology supports the premise that individuals struggle in such scenarios. 
%where they are presented with many choices. 
In his seminal work, \citet{schwartz2004paradox} puts forth the \emph{paradox of choice}: individuals incur increased anxiety when faced with too many alternatives, which often leads them to take a default action, defer, or not participate altogether~\citep{iyengar2000choice}. Recent literature has shown this phenomenon to hold specifically in the voting and social choice setting. \citet{cunow2021less, cunow2023too} experimentally show that even increasing the number of candidates from 3 to 6 leads voters to spend less effort learning candidates' policy positions and instead rely on arbitrary heuristics. This voter frustration is also evidenced in practice: incomplete ballots are quite common, which are often completely exhausted under elimination voting long before the final candidate is elected~\citep{Exhaustion}.  

These cognitive challenges are further exacerbated in contexts where the ``candidates'' are not politicians but \emph{opinions}, of which there may be very many. Prime examples of such contexts can be found in online platforms like \emph{Polis}~\citep{Polis}, \emph{Remesh},\footnote{\url{https://www.remesh.ai/}} \emph{All Our Ideas},\footnote{\url{https://allourideas.org/}} and \emph{Loomio},\footnote{\url{https://www.loomio.com/}} which facilitate deliberation, build consensus, and ultimately aggregate opinions on a specific topic. These platforms allow users to both submit opinions as free-form text and vote on submissions of others. Polis, for example, was deployed by the government of Taiwan to gauge sentiment on the regulation of ride-share apps, ultimately leading to new legislation \citep{horton2018simple}. Here, asking voters' opinions across all submissions can be far too time-consuming or downright infeasible. 
%to simply rank, say, their favorite 3 opinions submitted by other users is clearly infeasible, as simply reading all submissions is far too time-consuming. 
Instead, Polis makes a natural simplification by only showing a subset of opinions to each user. But what meaningful conclusions can be drawn from querying users over such limited data? And how should the platform select such queries? 

%Recent work by Halpern, Kehne, Procaccia, Tucker-Foltz, and W\"uthrich \citeyearpar{halpern2022representation} 
Recent work by \citet{halpern2022representation} studies this question in the context of approval votes, where the goal is to select a representative ``committee'' of size $k$. Voters from the population arrive randomly and can be presented with at most $t < m$ candidates (opinions) at a time, over which they can express their approval or disapproval. With sufficient arrivals, a committee selection algorithm can estimate the distribution of the population's approvals of any set of at most $t$ candidates. In an idealized query model, it is assumed that the algorithm exactly obtains this distribution in a single query. \citet{halpern2022representation} gives adaptive query algorithms to find committees satisfying the standard axioms of \emph{extended justified representation} and \emph{proportionality}. On the other hand, they also show information-theoretic lower bounds on the number of queries non-adaptive algorithms must make to guarantee a representative committee. 
%\sh{feels weird to mention one paper in such detail in the intro. Perhaps also mentioning \citet{bentert2020comparing} here would also be appropriate given the relevance of that model.}\daniel{I think we can potentially expand on it more in contributions, since it seems clsoer at a technical level, but further at a conceptual one.}

 Our paper extends this framework to consider ordinal preferences, with the more classic goal of selecting a single winner (rather than a committee) under a given rule. There is a distribution over rankings of the $m$ candidates representing the underlying voter preferences that is unknown to the algorithm. It may, however, \emph{query} randomly arriving individuals about any subset of $t$ candidates, where $t < m$; with suitable samples, the algorithm can determine the corresponding ranking distribution over this subset.\footnote{Clearly, having access to the exact distribution is strictly more informative than having to approximate it through repeated samples. Our negative results hold for this idealized setting and thus immediately apply to the weaker and more realistic query model.} This leads to the fundamental question: what is the set of voting rules that are implementable with such small-sized queries chosen by the platform? Conceptually, this question asks about the axiomatic implications of cognitive barriers to preference elicitation in social choice.  

%A \emph{query of size $t$} is allowed to ask for the induced distribution over any subset of $t$ candidates, where $t < m$.\footnote{Again, we emphasize that the application we have in mind sees each individual user asked about one set of $t$ candidates, e.g., when they arrive on the online platform. Clearly, learning the entire distribution of responses is a strictly more informative response, but it can be well approximated by repeated samples. In any case, our lower bounds for our idealized queries obviously apply to the weaker query notion as well.} The fundamental question we ask is, which voting rules are implementable with small-sized queries? Conceptually, this question asks about the axiomatic implications of cognitive barriers to preference elicitation in social choice.

\subsection{Contributions}

We begin our investigation with the ubiquitous plurality rule. Here, we find a surprisingly negative result: determining a plurality winner cannot be done using queries of any size $t < m$ (\Cref{thmPluralityCounter}). That is, even if one has access to the distribution of the population's rankings over every $m-1$ subset of the $m$ candidates, it is still impossible to correctly identify who received the most first-place votes. In fact, even a randomized algorithm can only correctly choose a plurality winner with probability $\frac{1}{m}$ in the worst case, i.e., there are instances where, no matter which queries an algorithm makes, it can do no better than picking a candidate uniformly at random. The proof follows from a novel construction of pairs of profiles that have different plurality winners but induce the same distribution on any subset of $m - 1$ candidates. \Cref{secConstruction} is dedicated to explaining this construction, which forms the basis for all of the impossibility results in this paper.

Plurality is just one example of a \emph{positional scoring rule}, whereby candidates receive points corresponding to their rank position in each ballot, with the winner being the candidate with the most aggregated points~\citep{Young75}. While plurality requires queries of size $m$, it is known that another positional scoring rule, the Borda count, only requires pairwise margins to determine a winner and, hence, can be computed with queries of size 2.
%which gives $m - 1$ points to the first choice candidate, $m - 2$ points to the second choice, and so on, with zero points for the last choice. % I think we need not explain borda rule while highlighting our contributions. This should be known to any reviewer here, and we do explain it in  detail in the main body. 
For general $t$-sized queries, one straightforward algorithm is to query every subset of size $t$ and give each candidate a certain number of points depending on which of the $t$ positions they appear.
%As for other scoring rules, one straightforward algorithm to compute them is to query every subset of size $t$, and give each candidate a certain number of points depending on which of the $t$ positions they appear.
Under a different query model, \citet{bentert2020comparing} give a family of scoring rules for each size $t$ for which this algorithm works.\footnote{In their model, each voter submits a ranking over a \emph{randomly} selected set of $t$ candidates.} A similar argument for the same families of rules holds for our model as well, which we prove in \Cref{lemRmtClosedForm} for completeness, along with a more detailed comparison between our work and theirs. We then proceed to our second main result (\Cref{thmCharacterizationPositional}) that shows that these rules are, in fact, the \emph{only} ones that can be implemented with queries of size $t$. We thus obtain a complete characterization of all computable positional scoring rules under a limited query model and also give visualizations of this space. Our analysis is then extended to the \emph{single transferrable vote (STV)} rule (also known as \emph{instant-runoff voting})
%\footnote{this is technically not a positional scoring rule.}
wherein we find a negative result akin to Plurality: algorithms with limited query size ($t < m$) cannot correctly implement this rule. All of these negative results again hold not only for deterministic algorithms but for randomized ones that are correct with probability strictly better than $\frac{1}{m}$ (i.e. better than randomly guessing).

We next use our characterization to study the rules that \emph{are} computable with small-sized queries. Specifically, for any scoring rule requiring $t^*$-sized queries, we lower bound the number of $t$-sized queries ($t^* \leq t < m$) needed to compute the winner under this rule (\Cref{thmCovering}). If $t$ and $t^*$ are treated as constants, then $\Theta(m^{t^*})$ queries are needed in the worst case. This asymptotic bound holds even for randomized algorithms that are correct with probability $\frac{1}{m} + \varepsilon$ for any constant $\varepsilon > 0$. If on the other hand an algorithm covers a $\delta$ fraction of all $t^*$-subsets, we give an upper bound on the success probability. Although this is not tight, we exactly determine the optimal success probability for Borda count with $m=3$ (\Cref{thmQC32}) with a surprisingly intricate construction, and leave open the general case.

%we consider the \emph{number} of queries of each given size $t$ that are necessary to implement positional scoring rules. For any rule that is implementable at all, we obtain matching upper and lower bounds (\Cref{thmCovering}), which are based on combinatorial quantities called \emph{covering numbers}.

\subsection{Related Work}

Our work contributes to a line of literature on the information-theoretic aspects of voting and elections, specifically on what can be accomplished with incomplete information. There are a variety of lenses through which to study this problem. One line of work considers the communication complexity of computing various voting rules~\citep{CS05}. Another studies when using incomplete votes can guarantee a candidate must or cannot be the winner, regardless of the missing information~\citep{KL05,XC11}. Others consider different ``approximation'' objectives such as minimax regret~\citep{LB11} and distortion~\citep{PR06}. For a more complete survey, see chapter 10 by \citet{Handbook}. 

More specific to our work are models where the partial information given is $t$-wise comparisons. The special case of $t=2$ corresponds to only being given pairwise comparisons. All information about pairwise comparisons can be summarized in a \emph{weighted tournament graph}, a widely studied object in social choice theory~\citep{Handbook}. For example, there is a classification of common voting rules into those that can be computed using just the tournament graph (such as \emph{Borda Count}, \emph{Minimax}, \emph{Kemeny}, and \emph{Copeland}) and those that cannot~\citep{Fish77}. Any of these weighted tournament solutions can trivially be computed with queries of size $t = 2$. Beyond information-theoretic results, there are even bounds on query complexity, such as how many queries to the tournament graph are needed to compute Condorcet winners~\citep{Pro08b}. Our paper is a natural extension of this literature to the realm of more powerful $t$-wise comparison queries for $t > 3$.

Related, but technically incomparable, is when voters reveal a ranking of their top $t$ candidates for a fixed value of $t$~\citep{oren2013efficient,FO14}. This was one of the two models studied by \citet{bentert2020comparing}. The other has voters revealing a ranking over a \emph{random} set of $t$ candidates. Positive results here translate to our model, although negative ones do not. We describe this connection and their results more in-depth in \Cref{subsec:characterization}.

Finally, note that we largely study information-theoretic impossibilities and thus do not focus on the randomized arrivals aspect; other works do consider the sample complexity of computing various rules~\citep{dey2015sample}; however, they assume randomly sampled voters reveal their complete preferences over all candidates.

% \jamie{TODO Talk about communication complexity, other things in the Handbook of Social Choice  chapter 10.}

% There has been much prior work $t = 2$. \jamie{TODO Talk about Ariel's Condorcet result, other papers like \citep{TournamentWinners}.} C1/C2 rule classification 

% Aside from the one paper of Bentert and Skowron \citep{bentert2020comparing} noted above, the authors are not aware of any other prior work studying queries of size $t \geq 3$. Our paper thus extends the literature on voting based on pairwise comparisons, introducing a new, more granular and quantitative measure of the information theoretic complexity of voting rules.

\section{Preliminaries}\label{secModel}

\subsection{Voter preferences}

For a positive integer $s$, let $[s] := \set{1, \ldots, s}$. A \emph{ranking} or \emph{preference} over a set $C$ of $m$ \emph{candidates} is a bijection $\sigma: [m] \to C$, where $\sigma(j)$ represents the $j$'th most-preferred candidate according to $\sigma$. We use the standard notation $a \succ_\sigma b$ to denote that $a$ is preferred to $b$ under $\sigma$, i.e., $\sigma^{-1}(a) < \sigma^{-1}(b)$, where $\sigma^{-1}$ is the inverse mapping from candidates to rankings. We write $\mathcal{L}(C)$ to denote the set of all $m!$ rankings over the candidates in $C$.  For a subset of candidates $S \subseteq C$, we write $\sigma|_S$ to denote the ranking $\sigma$ restricted to the candidates in $S$, i.e., $\sigma|_S \in \mathcal{L}(S)$, with $\sigma|_S(j)$ being the $j$'th most preferred among those in $S$ according to $\sigma$.

For a permutation over the candidates $\pi: C \to C$, we will write $\pi \circ \sigma$ for the ranking $\sigma$ permuted by $\pi$, i.e., $(\pi \circ \sigma)(j) = \pi(\sigma(j))$. Of particular interest will be permutations that swap a single pair of candidates. For this reason, for two candidates $a$ and $b$, we define $\pi^{ab}$ to be the \emph{$(a, b)$-transposition}, the permutation that swaps $a$ and $b$, i.e., $\pi^{ab}(a) = b$, $\pi^{ab}(b) = a$, and $\pi^{ab}(c) = c$ for all $c \ne a, b$. Further, we will write $\sigma^\atob$ for $\pi^{ab} \circ \sigma$.

A \emph{preference profile} (or simply a \emph{profile}) is a distribution $\bsigma$ over preferences $\mathcal{L}(C)$, representing the proportion of voters in the population that have each ranking. For example,
$$\Pr_{\sigma \sim \bsigma}[\sigma = a \succ b \succ c] = \frac15$$
means that $\frac15$ of the voters have the ranking $a \succ b \succ c$.\footnote{More often in social choice, a profile is a ranking assignment for a finite number of $n$ agents. The distributional definition is essentially equivalent insofar as voter identity is not important (as is the case for all rules we study) while having the additional benefit of making our query model and proof techniques easier to understand. However, we could have equivalently used the more traditional definitions, and all of our results would still hold.} We denote by $\profiledist(C)$ the space of all possible preference profiles over a candidate set $C$. For a profile $\bsigma \in \profiledist(C)$, $\sigma \sim \bsigma$ denotes sampling a preference $\sigma$ from distribution $\bsigma$. For a set of rankings $R \subseteq \mathcal{L}(C)$, we will also use the notation $\sigma \sim \Unif(R)$ to denote sampling a ranking $\sigma$ uniformly from $R$. We extend the restriction $\sigma|_S$, permutation $\pi \circ \sigma$, and transposition $\sigma^\atob$ operations from rankings to profiles in the natural way. More formally, for restrictions, we write $\bsigma|_S$ to denote the distribution $\bsigma$ when restricted to candidates in $S$, i.e., $\bsigma|_S$ is an element of $\profiledist(S)$ induced by sampling $\sigma \sim \bsigma$ and outputting $\sigma|_S$. For permutations, $\pi \circ \bsigma$ is the distribution induced by sampling $\sigma \sim \bsigma$ and outputting $\pi \circ \bsigma$. For transpositions, $\bsigma^\atob = \pi^{ab} \circ \bsigma$. For a set $S$, we also use the notation $\Unif(S)$ to denote the uniform distribution over elements of $S$.

\subsection{Voting rules}

A \emph{voting rule} $f$ maps preference profiles 
%$f: \profiledist(C) \to \mathcal{P}(C)$ from profiles 
to a set of winning candidates.
%\footnote{Mapping to a set instead of a single candidate is sometimes referred to as a \emph{social choice correspondence}.}
%\footnote{For example, we refer to candidates as a plurality winner or an STV winner}.
If a candidate $c$ is amongst the winners for a voting rule $f$, we refer to this candidate as an $f$-winner.

We will primarily focus on \emph{positional scoring rules} (or simply \emph{scoring rules}), which are a very practical and well-studied class.
These are parameterized by a \emph{scoring vector} $\balpha = (\alpha_1, \ldots, \alpha_m) \in \mathbb{R}^m$.\footnote{Often, scoring vectors are restricted to be nonnegative and nonincreasing, but, for our purposes, it will be more convenient to allow for arbitrary vectors.} Intuitively, a voter with ranking $\sigma$ gives $\alpha_1$ points to their first place candidate $\sigma(1)$, $\alpha_2$ points to their second place candidate $\sigma(2)$, and so on, with the winner of the profile being the candidate with the most aggregated points. Written in our distributional notation, the score of a candidate $c$ on profile $\bsigma$ is $\score^{\balpha}_{\bsigma}(c) := \E_{\sigma \sim \bsigma}[\alpha_{\sigma^{-1}(c)}]$, and the winning candidates on a profile $\bsigma$ are those with maximal score. Some common scoring rules include \emph{plurality}, parameterized by $(1, 0, \ldots, 0)$, \emph{veto}, parameterized by $(0, \ldots, 0, -1)$, and \emph{Borda count}, parameterized by $(m-1, m-2, \ldots, 0)$. Since the plurality score will come up quite frequently, we will write $\plu_{\bsigma}(c)$ instead of using the $\score^{\balpha}_{\bsigma}(c)$ notation, and note that this simplifies to $\plu_{\bsigma}(c) = \Pr_{\sigma \sim \bsigma}[\sigma(1) = c]$. For conciseness, we will use $\balpha$-winner, plurality winner, veto winner, and Borda winner to refer to $f$-winners of the scoring rule induced by $\balpha$, plurality, veto, and Borda count, respectively. Note that $\balpha$-winners are invariant under modifying $\balpha$ via translation or multiplication by a positive constant, e.g., veto-winners (with vector $(0, \ldots, 0, -1)$) coincide with both $(1, \ldots, 1, 0)$-winners and $(\frac1{m-1}, \ldots, \frac1{m-1}, 0)$-winners.

In addition to scoring rules, we will also consider the rule \emph{Single Transferable Vote (STV)}, which is defined as follows. For a profile $\bsigma$, if there is a single candidate, it returns that candidate. Otherwise, it chooses a candidate $c$ with \emph{minimal} plurality score, deletes them from the profile, and recurses on the rest. More formally, it chooses $c \in \argmin_{c'} \plu_{\bsigma}(c')$, and runs STV on $\bsigma | _{C \setminus \set{c}}$. This must eventually terminate, as a candidate is removed at each iteration. Further, for $m=2$, this coincides with plurality. Note that in some cases, there are ties for the minimal score candidate. Hence, we will say that $c$ is an STV winner if some sequence of valid eliminations results in $c$ being the winner.

\subsection{Query model}
We consider (possibly randomized) algorithms that are allowed to adaptively submit queries to the underlying distribution $\bsigma$. For a fixed parameter $t$, each query consists of a subset of candidates $Q \subseteq C$ with $|Q| \le t$ (referred to as a \emph{query of size $t$}, or a \emph{$t$-query} for short) and returns the distribution $\bsigma|_Q$. As mentioned previously, having access to the exact distribution is strictly more informative than one approximated by a random voter arriving (a voter $\sigma \sim \bsigma$ arrives, and the algorithm learns $\sigma|_Q$). All our impossibility results hold in this idealized setting and thus immediately apply to the more realistic one. 

Two profiles $\bsigma^1$ and $\bsigma^2$ are said to be \emph{$t$-indistinguishable} if for all subsets $Q \subseteq C$ with $|Q| \le t$, $\bsigma^1|_Q = \bsigma^2|_Q$. That is, regardless of whether the profile is $\bsigma^1$ or $\bsigma^2$, any query $Q$ with $|Q| \le t$ will have the same response.
%Thus for any algorithm, no matter whether the input is $\bsigma^1$ or $\bsigma^2$, on any query $Q$ with $|Q| \le t$, the response will be the same. 
Importantly, if $\bsigma^1$ and $\bsigma^2$ are $t$-indistinguishable, but $f(\bsigma^1) \cap f(\bsigma^2) = \emptyset$, then no $t$-query algorithm can always output an $f$-winner. Note that to check whether two profiles are $t$-indistinguishable, it suffices  to check only queries $Q$ of size exactly $t$, as if $Q' \subseteq Q$, then $\bsigma|_{Q'} = (\bsigma|_Q)|_{Q'}$, so $\bsigma^1|_Q = \bsigma^2|_Q$ implies $\bsigma^1|_{Q'} = \bsigma^2|_{Q'}$. Finally, notice that the special case of $2$-indistinguishable is equivalent to $\bsigma^1$ and $\bsigma^2$ having the same \emph{weighted tournament graph}, the complete-directed graph, where the nodes are candidates, and the weight on edge $(a, b)$ is the proportion of voters that prefer $a$ to $b$~\citep{Handbook}. In this sense, the collection of distributions $\{\bsigma|_Q\}_{Q \subseteq C : |Q| = t}$ are a generalization of the weighted tournament graph to arbitrary $t \ge 2$ (e.g., for $t=3$, rather than consisting of proportions of people that prefer $a$ to $b$, it consists of the proportions of people that prefer $a$ to $b$ to $c$ for all such combinations).

% \jamie{Some notes on notation. I purposely included $C$ in $\profiledist(C)$ because I use different candidate sets in a proof. I'm not too happy about the subtle difference in appearance between $\sigma$ and $\bsigma$. If we want to change the latter, we can redefine the command for $\bsigma$. We can also change the command for $\profiledist$; I'm not at all attached to this notation, though I think there should be some symbol for it. Finally, you'll notice my new definition of a voting rule always outputs a single candidate. I was thinking we could just break ties arbitrarily. It makes the notation simpler than if we have to worry about outputting sets of candidates.}\daniel{I think leave rules in $\mathcal{P}(C)$ form, and then be more explicit about what it means to be consistent with a rule in the theorem statements}.

% A voting rule $f$ is $t$-computable if, when there are $m$ candidates, it can be computed using queries of size $t$. \jamie{What does this mean?? See my August 27 email.}

% Alternatively, define an $(m, t)$-selection algorithm as one that works on $m$ candidates making queries of size $t$. Then say there is no $(m, t)$-selection algorithm that coincides with $f$.
% \sh{Any reason why we use $t$ for query size and not $q$ or $k$. Are we using these elsewhere? $t$ has a connotation of time/rounds.} \jamie{We used $t$ in our AAAI paper, so I'm happy to stick with $t$ for consistency.}

\section{An Indistinguishable Construction}\label{secConstruction}

We begin with a construction of two indistinguishable profiles which will be used throughout our technical results. 

\begin{lemma}\label{lem:construction}
    For any $m$ and pair of candidates $a, b \in C$, there is a profile $\bsigma$ such that (i) $\plu_{\bsigma}(a) \ne \plu_{\bsigma}(b)$  and (ii) $\bsigma$ and $\bsigma^\atob$ are $(m-1)$-indistinguishable.
\end{lemma}
\begin{proof}
Fix $m$, $a$, and $b$. 
% Next, suppose $M \ge 3$. 
%     \sh{TODO: consistently use $k$ or $t$ for queries} We first consider the $m=3$ case separately. \jamie{Is this just to build intuition? If so, I think we should move it to the Introduction. Unless it's really necessary to handle this case separately in the proof.} \sh{No this needs to be handled seperately}\daniel{I think $m=2$ needs to be handled seperately although it is completely trivial because $t=1$ gives no information, but isn't $m=3$ captured by the example below?} Consider profile $\bsigma^a$ as follows, with profile $\bsigma^b$ being identical except candidates $a$ and $b$ are swapped.
%     \begin{center}
%     \begin{tabular}{c|c}
%         $\nicefrac{1}{4}$ & $a \succ b \succ c$  \\
%         $\nicefrac{1}{4}$ & $a \succ c \succ b$\\
%         $\nicefrac{1}{4}$ & $b \succ c \succ a$\\
%         $\nicefrac{1}{4}$ & $c \succ b \succ a$.
%     \end{tabular}
%     \end{center}
%     Observe that in response to any query of size $t=2$, $Q = \{x, y\}$, in either $\bsigma^a$ or $\bsigma^b$, exactly half the voters prefer $x \succ y$ and the other half prefer $y \succ x$. However, note that $\plu_{\bsigma^a}(a) > \plu_{\bsigma^a}(b)$ and  $\plu_{\bsigma^a}(b) > \plu_{\bsigma^a}(a)$. In fact, candidates $a$ and $b$ are the plurality winners in their respective instances, which we make use of in \Cref{thmPluralityCounter}.
Let $C^{-ab} = C \setminus \set{a, b}$ be the set of remaining candidates. We have that $|C^{-ab}| = m - 2$. We define $\bsigma$ to be the distribution induced by the following random process. First, pick a set $S \subseteq C^{-ab}$ uniformly at random (i.e., each of the $2^{m - 2}$ sets with equal probability). 
Formally, let $\mathcal{P}(C^{-ab})$ denote the power set of $C^{-ab}$, and we will sample $S \sim \Unif(\mathcal{P}(C^{-ab}))$. Then choose a uniformly random ranking $\tau^S$ of the candidates in $S$ and a uniformly random ranking $\tau^{\overline{S}}$ of the candidates in $\overline{S} := C^{-ab} \setminus S$. Finally, if $|S|$ is even, output $\tau^S \succ a \succ b \succ \tau^{\overline{S}}$ and if $|S|$ is odd, output $\tau^S \succ b \succ a \succ \tau^{\overline{S}}$. This process is visually represented in \Cref{fig:profile-a}.

\begin{figure}[t]
     \begin{center}
     \begin{forest}
            for tree={
            grow=east, % tree direction
            parent anchor=east,
            child anchor=west, % edge anchors
            rounded corners, draw,
            align=center,
            edge={->},
            s sep = 9ex
            }
            [{Choose $S\sim  \Unif(\mathcal{P}(C^{-ab}))$},l sep = 4ex,anchor =east
            [{Choose $\tau^S \sim \Unif(\mathcal{L}(S))$\\ and $\tau^{\overline{S}} \sim \Unif(\mathcal{L}$$(\overline{S}))$},l sep = 16ex,anchor =west
            [
                $\tau^S \succ a \succ b \succ \tau^{\overline{S}}$,edge label={node[midway,sloped,above]{$|S|$ is even}}
            ]
            [
                $\tau^S \succ b \succ a \succ \tau^{\overline{S}}$,edge label={node[midway,sloped,above]{$|S|$ is odd}}
            ]
            ]]
        \end{forest}
    \end{center}
    \caption{A process inducing the distribution over rankings for $\bsigma$.}\label{fig:profile-a}
    \end{figure}
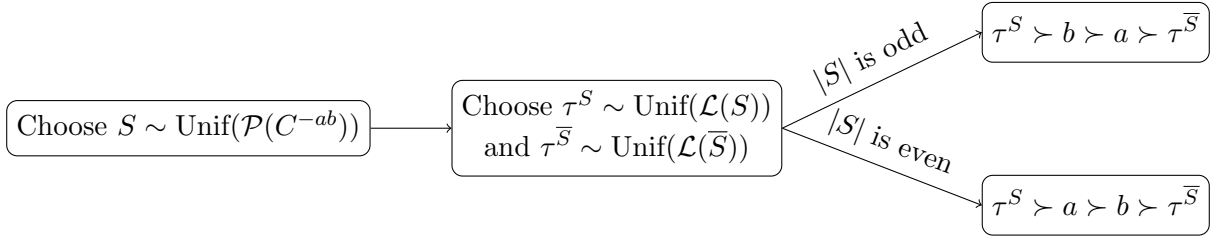

    First, observe that it is indeed the case that $\plu_{\bsigma}(a) \ne \plu_{\bsigma}(b)$. The only way that either $a$ or $b$ could be ranked first is if the set $S$ is empty. In that case, $|S|$ is even, so $a$ will be ranked first. This does happen with positive probability ($1/2^{m - 2}$), so the plurality score of $a$ is positive. However, this can never happen for candidate $b$.
    
    % A process defining the distribution of $\bsigma^a$ is shown in \Cref{fig:profile-a}. The distribution for $\bsigma^b$ is identical with $a$ swapped with $b$. Notice that this means that generating this profile simply involves doing the exact same process but swapping the order of $a$ and $b$ in the cases when $|S|$ is even or odd.
    
    %We first show that $a$ beats every other candidate by $2^{-k}$. Indeed, in the bottom branch, the only time a candidate \emph{not} in $C^{-ab}$ is ranked first is exactly when $S = \emptyset$. In this case, $a$ is ranked first, and this set is chosen with probability $1/2^{|C^{-ab}|} = 1/2^{k-1}$.  Therefore, $a$ has a plurality score of
    %\[
    %\frac{1}{k + 1} + \frac{k - 1}{k + 1} \cdot \frac{1}{2^{k - 1}} \ge \frac{1}{k + 1} + \frac{1}{2} \cdot  \frac{1}{2^{k - 1}} = \frac{1}{k + 1} + 2^{-k}, \] where the inequality holds because $t \ge 3$. Candidate $b$ is ranked first only in the second branch and hence has score $\frac{1}{k + 1}$. By symmetry, the $k - 1$ candidates of $C^{-ab}$ split the remaining $<\frac{k - 1}{k + 1}$ votes equally, and hence, each has a score of at most $\frac{1}{k + 1}$. Together, these imply that $a$ is strictly winning by $2^{-k}$.

    Next, we show that $\bsigma$ and $\bsigma^\atob$ are $(m-1)$-indistinguishable. Note that a ranking from $\bsigma^\atob$ can be sampled by running the process for $\bsigma$ but swapping the outcomes of the ``$|S|$ is even'' and ``$|S|$ is odd'' branches.
    %Note that to generate a ranking from $\bsigma^\atob$, we can run the process for $\bsigma$ simply swapping the branches for ``$|S|$ is even'' and ``$|S|$ is odd.'' 
    Fix any $Q \subseteq C$ with $|Q| = m-1$ and fix some ranking $\tau \in \mathcal{L}(Q)$. We want to show that observing $\tau$ is equiprobable under both $\bsigma|_Q$ and $\bsigma^\atob|_Q$.
    
    First, suppose $Q$ contains only one of $a$ or $b$. Without loss of generality, suppose it contains $a$. Since $Q$ does not contain $b$, if we run the process of \Cref{fig:profile-a} and pick $\tau^S$ and $\tau^{\overline{S}}$, regardless of whether we follow the ``$|S|$ is even'' or ``$|S|$ is odd'' branch the output when restricted to $Q$ is the same. This is due to $\bsigma$ and $\bsigma^\atob$ differing only in which branch to follow, and both branches are identical apart from the ordering of $a$ and $b$, which occur consecutively in both. Thus, outputting $\tau$ is equiprobable in both restricted profiles.
  
    Next, suppose $Q$ contains both $a$ and $b$. Again, since in any ranking of $\bsigma$ or $\bsigma^\atob$, $a$ and $b$ always appear adjacent to each other, if $a$ and $b$ are not adjacent in $\tau$, it occurs with probability $0$ in both profiles. As such, suppose they are adjacent in $\tau$, and without loss of generality, let $a \succ_{\tau} b$ (the other case is symmetric). Let $T$ be the set of candidates ranked above $a$ in $\tau$, and $L$ be the set of candidates ranked below $b$. Therefore, $Q = T \cup L \cup \set{a, b}$. Note to sample $\sigma$ with $\sigma|_Q = \tau$ under either $\bsigma$ or $\bsigma^\atob$, it must be the case that when selecting $S$, $S \cap Q = T$. Further, conditioned on this, the probability of getting both $T$ and $L$ in the order matching $\tau$ is simply $\frac{1}{|T|!|L|!}$. Finally, the order of $a$ and $b$ will match $\tau$ exactly when $|S|$ is even. Putting this together, we have that the 
    probability of sampling $\sigma$ with $\sigma|_Q = \tau$ under $\bsigma$ is exactly
    \[\frac{1}{|T|!|L|!} \Pr[(S \cap Q = T) \land (|S| \text{ is even})].
    \]
    For $\bsigma^\atob$, it is identical but with ``even'' switched with ``odd.'' Hence, to show equality, it suffices to show that:
    \[
        \Pr[(S \cap Q = T) \land (|S| \text{ is even})] = \Pr[(S \cap Q = T) \land (|S| \text{ is odd})].
    \]
    This is equivalent to showing that
    \[
        \Pr[|S| \text{ is even} \,|\, S \cap Q = T] = \Pr[|S| \text{ is odd} \,|\, S \cap Q = T].
    \]
    Let us now consider how to sample $S$ from the conditional distribution given $S \cap Q = T$. Note that $S$ satisfies this exactly when $T \subseteq S$ and $L \cap S \ne \emptyset$. Hence, to sample such an $S$, we can sample $S'$ uniformly from $C^{-ab} \setminus (T \cup L)$ and output  $S' \cup T$. By the assumptions that $|Q| = m - 1 < |C|$ and $\{a, b\} \subseteq Q$, we have that $C^{-ab} \setminus (T \cup L)$ is nonempty. It is known that when sampling a subset uniformly at random from a non-empty set, it is equiprobable whether it is of even or odd carnality.\footnote{Fix an element $x$. For any subset $S$ of the remaining elements, it is equally likely to pick $S$ and $S \cup \set{x}$. One of these has even parity and the other has odd.}  Hence, $|S'|$ is equally likely to be even or odd, which implies that  $|S|$ conditioned on $ S \cap Q = T$ is equally likely to be even or odd, as needed. 
    % Finally, consider the case of $m = 2$. In this case, $m - 1 = 1$, and note that querying a single candidate is degenerate in the sense that \emph{all} pairs of profile are $1$-indistinguishable. Therefore, we can simply choose $\bsigma$ to be the profile of all voters preferring $a \succ b$. \daniel{If we want to be super slick, I think this is technically captured in the $m \ge 3$ proof. Just the case where both $a$ and $b$ are contained in $Q$ will be vacuous as it can never happen.}
\end{proof}

\section{Uncomputable Voting Rules}\label{secScoringRules}

In this section, we prove that there is a family of voting rules that \emph{cannot} be computed using limited query sizes. In particular, we give a complete characterization of which positional scoring rules can be computed in our query model for each choice of $t$. In addition, we analyze the widely adopted STV rule. We begin, however, with two lemmas, which together give sufficient conditions for a scoring rule to \emph{not} be computable using limited queries. The  second will also be helpful for the STV impossibility.

\begin{lemma}\label{lemma:m_indist}
    Fix a vector $\balpha$, and suppose there exists a profile $\bsigma$ and two candidates $a$ and $b$ such that $\score^\balpha_{\bsigma}(a) \neq \score^\balpha_{\bsigma}(b)$ yet $\bsigma$ and $\bsigma^\atob$ are $t$-indistinguishable. Then, there exists a family of profiles, $\{\bsigma^c\}_{c \in C}$ such that all are $t$-indistinguishable from one another, but each candidate $c \in C$ uniquely maximizes $\score^\balpha_{\bsigma^c}$.
\end{lemma}
%    We will construct $m$ profiles, $\bsigma^c$ for each $c \in C$, such that all of them are $t$-indistinguishable from each other, yet $\score^{\balpha}_{\bsigma^c}(c) > \score^{\balpha}_{\bsigma^c}(c')$ for all $c' \ne c$.
\begin{proof}
Let $\bsigma$, $a$, and $b$ be the profile and candidates satisfying the lemma conditions.
     Fix a candidate $c$. We describe the distribution $\bsigma^c$ as follows. Assume without loss of generality that $\score^{\balpha}_{\bsigma}(a) > \score^{\balpha}_{\bsigma}(b)$. We first sample a permutation over the candidates $\pi$ uniformly at random. If $\pi(b) = c$, we return a ranking sampled from $\pi \circ \bsigma^\atob$; otherwise, we return a ranking sampled from $\pi \circ \bsigma$. A visual representation can be found in \Cref{fig:profile-mix}. 
    We will compare this constructed profile with respect to another, which we call $\bsigma^{\unif}$. The profile $\bsigma^{\unif}$ is constructed by picking a permutation $\pi$ uniformly at random and returning a sample $\sigma \sim \pi \circ \bsigma$ regardless of $\pi$, shown in \Cref{fig:profile-unif}.
    
    \begin{figure}[t]
     \begin{center}
     \begin{forest}
            for tree={
            grow=east, % tree direction
            parent anchor=east,
            child anchor=west, % edge anchors
            rounded corners, draw,
            align=center,
            edge={->},
            s sep = 8ex
            }
            [{Choose permutation $\pi$\\uniformly at random},l sep = 18ex,
            [
                Sample $\sigma \sim \pi \circ \bsigma$,edge label={node[midway,sloped,above]{$\pi(b) \ne c$}}
            ]
            [
                Sample $\sigma \sim \pi \circ \bsigma^\atob$,edge label={node[midway,sloped,above]{$\pi(b) = c$}}
            ]
            ]
        \end{forest}
    \end{center}
    \caption{A process inducing the distribution over rankings for $\bsigma^c$.}\label{fig:profile-mix}
    \end{figure}
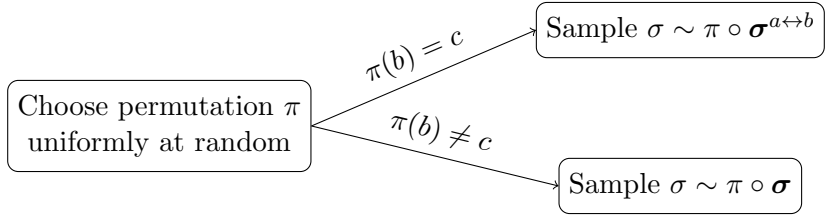
    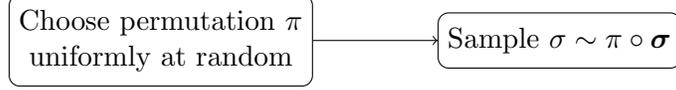
\begin{figure}[t]
     \begin{center}
     \begin{forest}
            for tree={
            grow=east, % tree direction
            parent anchor=east,
            child anchor=west, % edge anchors
            rounded corners, draw,
            align=center,
            edge={->},
            s sep = 8ex
            }
            [{Choose permutation $\pi$\\uniformly at random},l sep = 10ex,anchor=east
            [
                {Sample $\sigma \sim \pi \circ \bsigma$},anchor=west
            ]
            ]
        \end{forest}
    \end{center}
    \caption{A process inducing the distribution over rankings for $\bsigma^{\unif}$.}\label{fig:profile-unif}
    \end{figure}
    
    We first claim that $\bsigma^{\unif}$ is in fact the uniform distribution over $\mathcal{L}(C)$. Indeed, an equivalent way of sampling from $\bsigma^{\unif}$ is to sample in the reverse order, first sampling $\sigma \sim \bsigma$ and then applying a randomly selected permutation $\pi$ to $\sigma$. This makes it clear that $\bsigma^{\unif}$ is a mixture over uniform distributions and is hence uniform.
    
    We can think of both $\bsigma^c$ and $\bsigma^{\unif}$ as a mixture of $m!$ different profiles, the one associated with each choice of $\pi$. Abusing notation slightly, we will write $\bsigma^c_\pi$ and $\bsigma^{\unif}_\pi$ for the profile sampled when we had the permutation $\pi$. More concretely, 
    \[
        \bsigma^c_\pi = \begin{cases}
            \pi \circ \bsigma & \text{ if } \pi(b) \ne c\\
            \pi \circ \bsigma^\atob & \text{ if } \pi(b) = c
        \end{cases},
    \]
    while $\bsigma^{\unif}_\pi = \pi \circ  \bsigma$ for all $\pi$.
    
    Note that the scores $\score^{\balpha}_{\bsigma^{\unif}}(c')$ are equal for all $c'$ by symmetry. We will show both (i) $\score^{\balpha}_{\bsigma^c}(c) > \score^{\balpha}_{\bsigma^{\unif}}(c)$ while $\score^{\balpha}_{\bsigma^c}(c') \le \score^{\balpha}_{\bsigma^{\unif}}(c')$ for all $c' \ne c$ and (ii) $\bsigma^c$ is $t$-indistinguishable from $\bsigma^{\unif}$. The first shows that the $c$ is the unique score maximizer, and the second shows that all the constructed profiles are $t$-indistinguishable from each other since each is $t$-indistinguishable from $\bsigma^{\unif}$.
   
    For the first, consider the difference $\score^{\balpha}_{\bsigma^c}(c') - \score^{\balpha}_{\bsigma^{\unif}}(c')$ for an arbitrary candidate $c'$. Because scores are linear, we can split them across our mixture definitions to get
    \[
        \score^{\balpha}_{\bsigma^c}(c') = \E_{\pi}[\score^{\balpha}_{\bsigma^c_\pi}(c')] \text{ and }  \score^{\balpha}_{\bsigma^{\unif}}(c') = \E_{\pi}[\score^{\balpha}_{\bsigma^{\unif}_\pi}(c')].
    \]
    Plugging this into the difference, by linearity of expectation, we have that
    \[
        \score^{\balpha}_{\bsigma^c}(c') - \score^{\balpha}_{\bsigma^{\unif}}(c') = \E_{\pi}[\score^{\balpha}_{\bsigma^c_\pi}(c') -  \score^{\balpha}_{\bsigma^{\unif}_\pi}(c')].
    \]
    Now, for any $\pi$ with $\pi(b) \ne c$, the difference inside the expectation is $0$ because $\bsigma^c_\pi = \bsigma^{\unif}_\pi$. Fix some $\pi$ such that $\pi(b) = c$. In this case, $\bsigma^c_\pi = \pi \circ \bsigma^{\atob}$ while $\bsigma^{\unif}_\pi = \pi \circ \bsigma$. Note that \[\score^{\balpha}_{\pi \circ \bsigma^{\atob}}(c') = \score^{\balpha}_{\bsigma^{\atob}}(\pi^{-1}(c')) = \score^{\balpha}_{\bsigma}(\pi^{ab}(\pi^{-1}(c'))) \]
    where the first equality moves the application of $\pi$ and the second does the same, using the fact that $\bsigma^\atob$ is simply achieved by the $\pi^{ab}$ permutation (which is its own inverse). We also have:
    \[\score^{\balpha}_{\pi \circ \bsigma}(c') =  \score^{\balpha}_{\bsigma}(\pi^{-1}(c)).\]
    Hence, this difference is only nonzero if $\pi^{ab}(\pi^{-1}(c')) \ne\pi^{-1}(c')$, i.e., $\pi^{-1}(c') \in \set{a, b}$. If $\pi^{-1}(c') = a$, then the difference is $\score^{\balpha}_{\bsigma}(b) - \score^{\balpha}_{\bsigma}(a)$, while if $\pi^{-1}(c') = b$, then the difference is $\score^{\balpha}_{\bsigma}(a) - \score^{\balpha}_{\bsigma}(b)$. By the lemma assumptions, $\score^{\balpha}_{\bsigma}(a) > \score^{\balpha}_{\bsigma}(b)$; hence, the difference is positive in the first case and negative in the second. To summarize, $\score^{\balpha}_{\bsigma^c_\pi}(c') -  \score^{\balpha}_{\bsigma^{\unif}_\pi}(c')$ is nonzero only when $\pi(b) = c$ and either $\pi(b) = c'$ (in which case it is positive) or $\pi(a) = c'$ (in which case it is negative). From this description, we see that when $c' = c$, this can \emph{only} take on positive values (and does whenever $\pi(b) = c$), and when $c' \ne c$, this can \emph{only} take on negative values (and does whenever $\pi(b) = c$ and $\pi(a) = c'$). Hence, $\score^{\balpha}_{\bsigma^c}(c') - \score^{\balpha}_{\bsigma^{\unif}}(c')$ is positive when $c'=c$ and negative when $c' \ne c$, as needed.
    
    To complete the proof, we show that $\bsigma^c$ and $\bsigma^\unif$ are $t$-indistinguishable. To that end, fix $Q \subseteq C$ of size $t$. The key fact we will use is that if $\bsigma^1$ and $\bsigma^2$ are $t$-indistinguishable, then $\pi \circ \bsigma^1$ and $\pi \circ \bsigma^2$ are $t$-indistinguishable for all $\pi$. This implies that $(\pi \circ \bsigma^\atob)|_Q = (\pi \circ \bsigma)|_Q$ for all $\pi$. Therefore, both $\bsigma^c|_Q$ and $\bsigma^\unif|_Q$ are mixtures over the exact same $m!$ distributions, and are hence equal, as needed.
\end{proof}

\begin{lemma}\label{lem2ToM}
    Fix a voting rule $f$ and suppose there are $m$ profiles that are all $t$-indistinguishable, but each has a distinct singleton $f$-winner. Then, for all (possibly randomized) algorithms $A$ which on input profile $\bsigma$ can make queries of size at most $t$ to $\bsigma$ and output a candidate, there is a profile $\bsigma^*$ with a unique $f$-winning candidate $c$, such that the probability $A$ outputs $c$ on $\bsigma^*$ is at most $1/m$.
\end{lemma}
\begin{proof}
    Let $\{\bsigma^c\}_{c \in C}$ denote the set of $m$ profiles that are all $t$-indistinguishable, and let $c$ be the $f$-winner on profile $\bsigma^c$. Note that an algorithm run on any of these profiles will receive the exact same responses to queries. Hence, its output must be identical for all of them. There must be some candidate $a^*$ which it outputs with probability at most $1/m$, and hence, $\bsigma^{a^*}$ satisfies the desired properties.
\end{proof}

Taken together, these two lemmas outline sufficient conditions wherein limited query algorithms cannot compute the winner. Combined with the construction presented in \Cref{lem:construction}, it allows us to immediately conclude the following result about the impossibility of computing a plurality winner with any restricted query size.

\begin{theorem}\label{thmPluralityCounter}
    For any number of candidates $m \geq 2$, for all $t < m$, no randomized $t$-query algorithm can always output a plurality winner with probability more than $1/m$.
\end{theorem}
\begin{proof}
   In \Cref{lem:construction}, we proved the existence of a profile $\bsigma$ that has distinct plurality scores for two candidates $a, b$, and is $m-1$ indistinguishable from its transposed profile $\bsigma^\atob$. As such, we can apply \Cref{lemma:m_indist} and show the existence of $m$ profiles with distinct plurality winners that are all $m-1$ indistinguishable. Applying \Cref{lem2ToM} on these $m$ profiles directly gives us the desired result for plurality.
\end{proof}

While the result for plurality follows immediately, the question of computing arbitrary scoring rules with a limited query size requires a more involved approach. This is tackled next.

\subsection{Characterization of scoring rules}\label{subsec:characterization}
We now consider an arbitrary scoring vector $\balpha$ and determine the exact query size needed to compute an $\balpha$-winner. Fix $1 \leq t \leq m$, and fix a candidate set $C$ of size $m$. The following notation will be convenient for discussing arbitrary positional scoring rules. For any preference profile $\bsigma \in \profiledist(C)$ and candidate $c \in C$, we define $\pos(\bsigma, c) \in \rr^m$ to be the vector of positional occurrences of candidate $c$ across all rankings in profile $\bsigma$. More formally, for each $j \in [m]$,
$$\pos(\bsigma, c)_j := \Pr_{\sigma \sim \bsigma}[\sigma(j) = c].$$
Thus, a positional scoring rule is a voting rule parameterized by scoring vector $\balpha$ which selects a candidate $c$ maximizing $\score^\balpha_{\bsigma}(c) = \balpha \cdot \pos(\bsigma, c)$. Now for each $k \in [t]$, consider the scoring vector $\balpha^k \in \rr^m$ given by:
\begin{equation*}
    \alpha^k_j = \binom{j - 1}{k - 1}\binom{m - j}{t - k}.
\end{equation*}
We define the subspace spanned by these $k$ vectors as $R_{m, t}$. Formally,
$$R_{m, t} := \spn(\balpha^1, \balpha^2, \dots, \balpha^t) \subseteq \rr^m.$$

%We say that a voting rule $f: \profiledist(C) \to C$ is a \emph{positional scoring rule} if there is some vector $v(f) \in \rr^m$ such that, for any $\bsigma \in \profiledist(C)$,
%$$f(\bsigma) \in \argmax_{c \in C} (v(f) \cdot v(\bsigma, c)),$$
%where $\cdot$ denotes the dot product. We call the term $v(f) \cdot v(\bsigma, c)$ the \emph{score} of $c$.

%\sh{If we mention the span here, just want to briefly mention that Borda is in the span. Just as a grounding point to interpret}
To build some intuition for this space, it can be shown that for $t \geq 2$, $R_{m,t}$ contains the commonly used Borda score, corresponding to $\balpha = (m-1, m-2, \dots, 0)$ (see \Cref{figTetrahedron} for a visualization). We next show that any scoring rule in this space can be computed with $t$-sized queries and give a constructive algorithm. This is essentially shown in Theorem~1 of \citet{bentert2020comparing} but under a different model. For completeness, we present the proof for our setting below:

\begin{lemma}[Theorem 1 of \citet{bentert2020comparing}]\label{lemRmtClosedForm}
For any $t \le m$ and $\balpha \in R_{m, t}$, given $\bsigma \in \Pi(C)$ and a candidate $c \in C$, it is possible to compute $\score^{\balpha}_\bsigma(c)$ with queries of size $t$.
\end{lemma}

\begin{proof}
We will show that for each $k \le t$, it is possible to compute $\score^{\balpha^k}_\bsigma(c)$. For any $\balpha \in R_{m, t}$, since $\balpha = \lambda_1 \balpha^1 + \cdots + \lambda_t \balpha^t$ for some scalars $\lambda_1, \ldots, \lambda_t$, by linearity, $\score^{\balpha}_\bsigma(c) = \lambda_1 \score^{\balpha^1}_\bsigma(c) + \cdots + \lambda_t \score^{\balpha^t}_\bsigma(c)$. Hence, as long as we can compute the score for each of the basis vectors, we can do so for any vector in the span.

    Fix a candidate $c$ and index $j \in [m]$, and let $\sigma$ be any ranking such that $\sigma(j) = c$. Since there are $m$ candidates in $C$, there are $\binom{m}{t}$ possible sets $S \subseteq C$ of size $t$. The number of such subsets $S$ for which the restricted ranking $\sigma|_S$ puts $c$ in a position $k$ is given by $\binom{j - 1}{k - 1}\binom{m - j}{t - k}$, since $S$ must contain the special candidate $c$, along with $k - 1$ of the $j - 1$ candidates ranked above $c$ in $\sigma$, and $t - k$ of the $m - j$ candidates ranked below $c$ in $\sigma$. Thus, the probability that $\sigma|_S(k) = c$ for a uniformly random subset $S$ of size $t$ is
    $$\frac{\binom{j - 1}{k - 1}\binom{m - j}{t - k}}{\binom{m}{t}}.$$
    
    Consider the following algorithm. For any input preference profile $\bsigma \in \profiledist(C)$ and candidate $c \in C$, we first compute the probability that, if we draw a set $S$ of $t$ distinct candidates uniformly at random from $C$ and draw a preference $\sigma \sim \bsigma$, candidate $c$ will be at position $k$ in the ranking $\sigma$ restricted to $S$ - i.e the event $\sigma|_{s}(k) = c$. Clearly, this probability can be determined with queries of size $t$ by brute-forcing over all subsets $S$ of size $t$. We then output this probability multiplied by $\binom{m}{t}$. Observe that we may write the output of this algorithm as
    \begin{align*}
        \binom{m}{t}\Pr_{\substack{S \subseteq C,\ \abs{S} = t\\\sigma \sim \bsigma}}[\sigma|_S(k) = c] &= \binom{m}{t}\sum_{j = 1}^t \Pr_{\sigma \sim \bsigma}[\sigma(j) = c] \Pr_{\substack{S \subseteq C,\ \abs{S} = t\\\sigma \sim \bsigma}}[\sigma|_S(k) = c \suchthat \sigma(j) = c]\\
        &= \binom{m}{t}\sum_{j = 1}^t \pos(\bsigma, c)_j \frac{\binom{j - 1}{k - 1}\binom{m - j}{t - k}}{\binom{m}{t}} \stext{by the calculation above}\\
        &= \sum_{j = 1}^t \balpha^k_j \pos(\bsigma, c)_j = \balpha^k \cdot \pos(\bsigma, c) = \score^{\balpha^k}_{\bsigma}(c).\qedhere
    \end{align*}
\end{proof}

We now move to our main result which generalizes \Cref{thmPluralityCounter} by proving that $R_{m,t}$is exactly the space of all scoring rules computable with limited queries of size $t$, thus giving a complete characterization. 

\begin{theorem}\label{thmCharacterizationPositional}
    For any number of candidates $m \geq 2$, any $1 \leq t \leq m$, and any vector $\balpha \in \rr^m$ 
    %\daniel{$\balpha^0$? I'm also ok with just calling it $\balpha$ for the theorem statement since it's a general vector}\sh{Agree - $\balpha^0$ is a bit confusing notationally since basis vectors also use superscript. Prefer just $\balpha$}:
    \begin{enumerate}
        \item\label{itmCharacterizationPositive} If $\balpha \in R_{m, t}$ then there a $t$-query algorithm that always outputs a candidate $c$ maximizing $\score^{\balpha}_{\bsigma}(c)$ on any input profile $\bsigma$.\item\label{itmCharacterizationNegative} If $\balpha \notin R_{m, t}$, then no randomized $t$-query algorithm can always output an $\balpha$-winner with probability more than $\frac1m$.
    \end{enumerate}
\end{theorem}

\begin{figure}[t]
     \begin{center}
     \begin{forest}
            for tree={
            grow=east, % tree direction
            parent anchor=east,
            child anchor=west, % edge anchors
            rounded corners, draw,
            align=center,
            edge={->},
            s sep = 8ex
            }
            [{Sample $\sigma \sim \bsigma$},l sep = 10ex,anchor=east
            [
                {Output $c^2_1 \succ \cdots c^2_{i - 1} \succ \sigma \succ c^2_i \succ \cdots \succ c^2_{|C^2|}$},anchor=west
            ]
            ]
        \end{forest}
    \end{center}
    \caption{A process inducing $\bsigma^i$. This is parameterized by two disjoint sets, $C^1$ and $C^2$ and two candidates $a, b \in C^1$. The profile $\bsigma \in \Pi(C^1)$ satisfies the conditions of \Cref{lem:construction} with $C_1, a,$ and $b$. The indexing $c^2_1 \succ \cdots \succ c^2_{|C^2|}$ is an arbitrary order of the candidates in $C^2$.}\label{fig:sigmai}
    \end{figure}
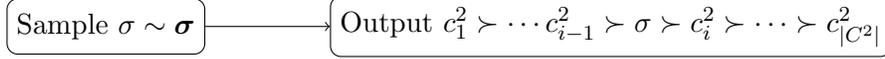

Statement (\ref{itmCharacterizationPositive}) is the easier part and follows immediately from \Cref{lemRmtClosedForm}. Our main focus here is statement (\ref{itmCharacterizationNegative}), which leverages the construction from \Cref{secConstruction}. We define a sequence of $m - t$ profiles as follows. We partition the candidate set $C$ into two disjoint pieces, a set $C_1$ of size $t + 1$ with two distinguished candidates $a, b \in C_1$, and a set $C_2$ of size $m - t - 1$. Let $\bsigma \in \profiledist(C_1)$ be a profile satisfying the conditions of \Cref{lem:construction} with $C_1$, $a$ and $b$, i.e., $\plu_\bsigma(a) \ne \plu_\bsigma(b)$ and $\bsigma$ and $\bsigma^{a \leftrightarrow b}$ are $t$-industinguishable. Let $c^2_1 \succ \cdots \succ c^2_{|C_2|}$ be an arbitrary order of the candidates in $C^2$. For each $i \in [m - t]$, we extend $\bsigma$ to a profile on all $m$ candidates $\bsigma^i \in \profiledist(C)$ by inserting it in the $C^2$ order after the first $i - 1$ candidates. A visual representation of this can be found in \Cref{fig:sigmai}. Note that, for each $i \in [m - t]$, we clearly have that $\bsigma^i$ is $t$-indistinguishable from $(\bsigma^i)^\atob$ since $\bsigma$ is $t$-indistinguishable from $\bsigma^\atob$, as for any $t$-query query $Q$, the candidates from $C_2$ are all in the same order, and $\bsigma|_{Q \cap C_1} = \bsigma^\atob|_{Q \cap C_1}$ because $|Q \cap C_1| \le t$. We now show for any vector not in the span, the score for $a$ and $b$ on one of these $t$-indistinguishable profiles is not the same.

\begin{lemma}\label{lemCharacterizationForwardDirection}
    For any $\balpha \notin R_{m, t}$, there exists some $i \in [m - t]$ such that $\score^{\balpha}_{\bsigma^i}(a) \neq \score^{\balpha}_{\bsigma^i}(b)$.
\end{lemma}

\begin{proof}
    For each $i \in [m - t]$, we define the vector
    $$s^i := \pos(\bsigma^i, a) - \pos(\bsigma^i, b) \in \rr^m.$$
    Suppose toward a contradiction that candidates $a$ and $b$ have the same scores in each $\bsigma^i$ according to $\balpha$. This implies that
    \[
        \balpha \cdot \pos(\bsigma^i, a) = \balpha \cdot \pos(\bsigma^i, b).
    \]
    
    Since the score under each basis vector  $\score^{\balpha^k}_{\bsigma^i}$ for $k \in [t]$ can be computed with queries of size $t$ by \Cref{lemRmtClosedForm} and each $\bsigma^i$ is $t$-indistinguishable from $(\bsigma^i)^\atob$, we know that $a$ and $b$ must have the same scores in $\score^{\balpha^k}_{\bsigma^i}$ for each $\bsigma^i$ and $k \in [t]$ as well. In other words, for each $0 \leq k \leq t$, and for each $i \in [m - t]$, we have
    $$\balpha^k \cdot \pos(\bsigma^i, a) = \balpha^k \cdot \pos(\bsigma^i, b).$$
    We can therefore equivalently write for all $\balpha' \in \set{\balpha, \balpha^1, \balpha^2, \dots, \balpha^t}$ and all $i \in [m - t]$,
    $$\balpha' \cdot s^i = 0.$$
    Consider the following two subspaces of $\rr^m$:
    \begin{align*}
        R &:= \spn(\balpha, \balpha^1, \balpha^2, \dots, \balpha^t)\\
        S &:= \spn(s^1, s^2, \dots, s^{m - t})
    \end{align*}
    What we have just shown is that the vectors generating $R$ are each orthogonal to the vectors generating $S$, so the spaces are orthogonal to each other. Therefore, $\dim(R) + \dim(S) \leq m$. To obtain a contradiction, we will show that $\dim(R) = t + 1$ and $\dim(S) = m - t$.

    First consider $R$. For each $k \in [t]$, observe that the $j\tth$ entry of $\balpha^k$ is zero for all $j < k$, since $\binom{j - 1}{k - 1} = 0$. On the other hand, the $k\tth$ entry of $\balpha^k$ is
    $$\binom{k - 1}{k - 1}\binom{m - k}{t - k} = \binom{m - k}{t - k} > 0$$
    because $k \leq t \leq m$. Thus, arranging the vectors $\balpha^1, \balpha^2, \dots, \balpha^t$ as the rows of a matrix, we have a triangle of zeros below nonzero diagonal entries of the form
    $$M =
    \begin{pmatrix}
        \balpha^1 \\
        \balpha^2 \\
        \vdots \\
        \balpha^t
    \end{pmatrix}
    =
    \begin{pmatrix}
        \balpha^1_1 > 0 & \balpha^1_2 & \cdots & \balpha^1_t \\
        0 & \balpha^2_2 > 0 & \cdots & \balpha^2_t \\
        \vdots & \vdots & \ddots & \vdots \\
        0 & 0 & \cdots & \balpha^t_t > 0
    \end{pmatrix}$$
    Clearly, $M$ has full rank, so the $\balpha^k$ vectors are all linearly independent. Furthermore, $\balpha$ is independent of all of these $t$ basis vectors, since we are assuming $\balpha \notin R_{m, t}$. Hence, the $t + 1$ vectors generating $R$ are independent, so $R$ has dimension $t + 1$.

    Now consider $S$. For each $i \in [m - t]$, there is no way that either of the two special candidates $a$ and $b$ could be ranked above position $i$ in any preference $\sigma$ in the support of $\bsigma^i$, since $a$ and $b$ belong to $C_1$, not $C_2$. Thus, for each $j < i$, we have
    $$s^i_j = \pos(\bsigma^i, a)_j - \pos(\bsigma^i, b)_j = 0 - 0 = 0.$$ 
    On the other hand, $a$ and $b$ occur exactly at position $i$ in $\bsigma^i$ with the same probability that they occur first in $\bsigma$ satisfying \Cref{lem:construction}. By \Cref{lem:construction}, the probability of $a$ occurring at the first position is different than the probability of $b$ occurring at the first position. In other words,
    $$s^i_i = \pos(\bsigma^i, a)_i - \pos(\bsigma^i, b)_i = \pos(\bsigma, a)_1 - \pos(\bsigma, b)_1 = \plu_{\bsigma}(a) - \plu_{\bsigma}(b) \neq 0.$$
    Thus, by the same triangle-of-zeros argument as before, we conclude that the $s^i$ vectors are independent, so $\dim(S) = m - t$.
\end{proof}

\begin{proof}[Proof of \Cref{thmCharacterizationPositional}]
    First suppose $\balpha \in R_{m, t}$. By \Cref{lemRmtClosedForm}, on input $\bsigma$, we can compute $\score^{\balpha}_\bsigma(c)$ for each candidate $c$ using queries of size $t$. Therefore, we can output a candidate maximizing this score.
    % Then $\balpha \in \spn(\balpha^1, \balpha^2, \dots, \balpha^t)$, so we may write
    % $$\balpha = \sum_{k = 1}^t{\beta_k} \balpha^k.$$
    % By \Cref{lemRmtClosedForm}, the value of each candidate $c$'s score $\balpha^k \cdot \pos(\bsigma, c)$ may be computed with queries of size $t$, so it follows that the value of each
    % $$\balpha \cdot \pos(\bsigma, c) = \sum_{k = 1}^t{\beta_k} (\balpha^k \cdot \pos(\bsigma, c))$$
    % can as well.

    On the other hand, suppose $\balpha \notin R_{m, t}$. Then let $\bsigma^i$ be as in \Cref{lemCharacterizationForwardDirection}. Since $\score^{\balpha}_{\bsigma^i}(a) \neq \score^{\balpha}_{\bsigma^i}(b)$, yet $\bsigma^i$ is indistinguishable from $(\bsigma^i)^\atob$, the conclusion follows immediately from \Cref{lemma:m_indist,lem2ToM}.
\end{proof}

\begin{figure}
  \begin{minipage}[c]{0.61\textwidth}
    \includegraphics[width=\textwidth]{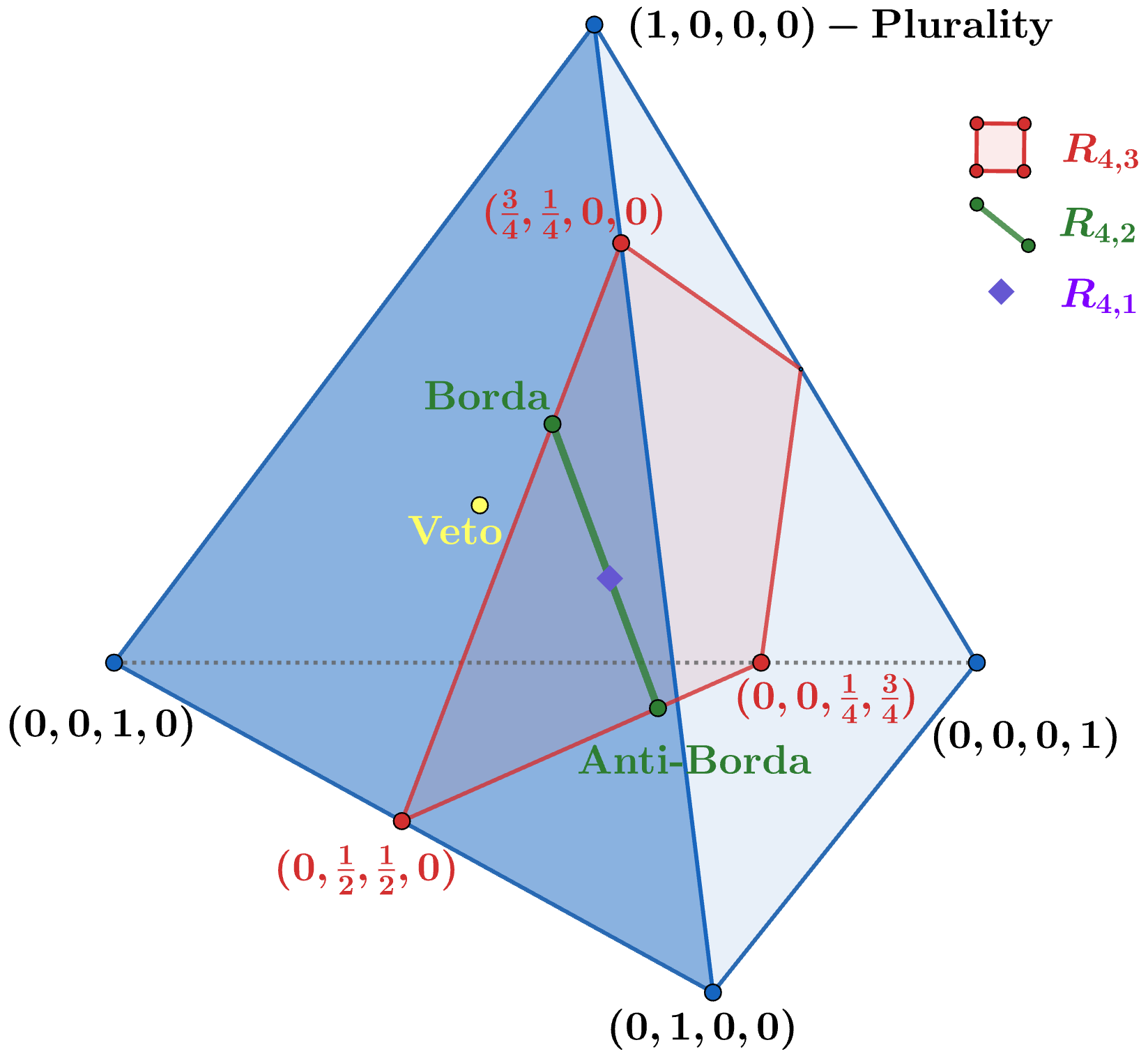}
  \end{minipage}\hfill
  \begin{minipage}[c]{0.35\textwidth}
    \caption{The space of positional scoring rules for $m = 4$ candidates. The corners represent rules that give all weight to a single position, with the top being Plurality. The red kite-shaped region is the 2-dimensional subspace $R_{4, 3}$ spanned by the vectors $\balpha^1 = (3, 1, 0, 0)$, $\balpha^2 = (0, 1, 1, 0)$, and $\balpha^3 = (0, 0, 1, 3)$, which are normalized in the simplex as the three red points at the corners of the kite. As we have shown, this subspace does not contain Plurality. Nested within this subspace is $R_{4, 2}$, the green 1-dimensional subspace spanned by the Borda and Anti-Borda scoring vectors. The purple point at the middle of this line is the trivial voting rule that gives every candidate the same score, which is the only element of the 0-dimensional subspace $R_{4, 1}$.}
    \label{figTetrahedron}
  \end{minipage}
\end{figure}

% \begin{figure}
%     \includegraphics[width=.6\textwidth]{tetra_scoring_plot_legend.png}
%     \caption{The space of positional scoring rules for $m = 4$ candidates. The corners represent rules that give all weight to a single position, with the top being Plurality. The red kite-shaped region is the 2-dimensional subspace $R_{4, 3}$ spanned by the vectors $\balpha^1 = (3, 1, 0, 0)$, $\balpha^2 = (0, 1, 1, 0)$, and $\balpha^3 = (0, 0, 1, 3)$, which are normalized in the simplex as the three red points at the corners of the kite. As we have shown, this subspace does not contain Plurality. Nested within this subspace is $R_{4, 2}$, the green 1-dimensional subspace spanned by the Borda and Anti-Borda scoring vectors. The purple point at the middle of this line is the trivial voting rule that gives every candidate the same score, which is the only element of the 0-dimensional subspace $R_{4, 1}$.}
%     \label{figTetrahedron}
% \end{figure}

%\daniel{It's only 2d in this projection, in general, it's 3d (same with $R_{4, 2}$ and $R_{4, 1}$). Not sure if we should make this distinction or not}\sh{addressed in footnote},

To make this space of computable scoring rules interpretable, we visualize the subspaces $R_{m, t}$ for $m = 4$. Since translating a scoring vector by a constant and scaling by a positive value do not affect the induced rule, all rules in $\mathbb{R}^4$ can be normalized such that they are contained within the 3-dimensional simplex (a tetrahedron). For instance, the scoring vector $(3, 2, 1, 0)$ for Borda is equivalent to $(\frac12, \frac13, \frac16, 0)$ while the one for veto $(0, 0, 0, -1)$ corresponds to $(\frac13, \frac13, \frac13, 0)$. \Cref{figTetrahedron} depicts this 3-simplex of scoring rules for $m = 4$ candidates and highlights its intersection with the subspaces $R_{4, 3}$, $R_{4, 2}$, and $R_{4, 1}$ along with other rules of interest.\footnote{Although as a subspace of $\mathbb{R}^m$, $R_{m,t}$ is $t$-dimensional, when we restrict to the simplex, we lose a dimension. Hence $R_{4,4}$ becomes 3-dimensional, $R_{4,3}$ becomes 2-dimensional and so on.} Note that $R_{4,4}$ corresponds to the whole simplex.

\subsection{Single Transferrable Vote}
Next, we consider the Single Transferrable Vote (STV) which cannot be parameterized by a scoring vector. Nonetheless, similar to plurality, we find a strong negative result about its computability with any limited-sized queries.

\begin{theorem}
    For any number of candidates $m \geq 2$, for all $t < m$, no randomized $t$-query algorithm can always output an STV winner with probability more than $\tfrac{1}{m}$. 
\end{theorem}
\begin{proof}
    Observe that when $m=2$, the STV winner is equivalent to the plurality winner, so this is directly implied by \Cref{thmPluralityCounter}. Fix $m \geq 3$. We would like to apply \Cref{lem2ToM}, and to do so, we will construct $m$ profiles $\{\bsigmaSTV^{c_1}, \dots, \bsigmaSTV^{c_m}\}$ that are all $(m-1)$-indistinguishable, but the STV winner on profile $\bsigmaSTV^{c_i}$ is candidate $c_i$. The construction of such profiles is as follows. Let $\alpha = (-1, 0, \dots, 0)$, the negation of the plurality score vector. Since on any profile $\bsigma$ and any candidate pair $a,b$, we have that $\plu_{\bsigma}(a) \ne \plu_{\bsigma}(b) \iff sc_a^{\alpha}(\bsigma) \ne sc_b^{\alpha}(\bsigma)$, we can use \Cref{lem:construction,lemma:m_indist} to obtain a set of ($m-1$)-indistinguishable profiles $\{\bsigma^c\}_{c \in M}$, where candidate $c$ is the unique $\alpha$ score maximizer on profile $\bsigma^c$; correspondingly, it is the unique plurality minimizer on profile $\bsigma^c$ due to the choice of $\alpha$.
    
    Next, fix a directed cycle of the candidates $D := c_1 \rightarrow c_2 \rightarrow \cdots \rightarrow c_m \rightarrow c_1$. Let $R$ be the set of all rankings such that the first and last candidates are consecutive in the cycle, i.e., $R = \set{\sigma \in \mathcal{L}(C) \mid (\sigma(1), \sigma(m)) \in D}$.  Choose $\varepsilon > 0$ such that $\varepsilon < \frac{1 - \varepsilon}{m(m - 1)}$ ($\varepsilon = \frac{1}{m^2}$ will do). Fix a candidate $c$, and let $\text{next}(c)$ be the subsequent candidate in the cycle, i.e., the unique candidate $c'$ such that $(c, c') \in D$. Define $\bsigmaSTV^{c}$ as follows: with probability $\varepsilon$, output a sample $\sigma$ from $\bsigma^{\text{next}(c)}$, with remaining probability $1 - \varepsilon$, select $\sigma$ from $\Unif(R)$. A visual representation of this can be found in \Cref{fig:profile-stv}.
    
    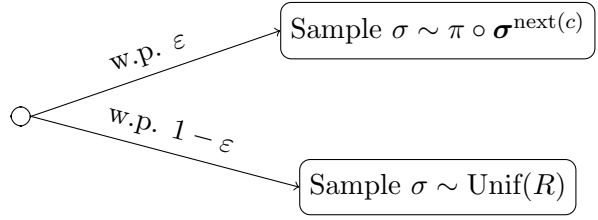
\begin{figure}[t]
     \begin{center}
     \begin{forest}
            for tree={
            grow=east, % tree direction
            parent anchor=east,
            child anchor=west, % edge anchors
            rounded corners, draw,
            align=center,
            edge={->},
            s sep = 8ex
            }
            [{},l sep = 20 ex,
             [
                Sample $\sigma \sim \Unif(R)$,edge label={node[midway,sloped,above]{w.p. $1- \varepsilon$}}
            ]
            [
                Sample $\sigma \sim \pi \circ \bsigma^{\text{next}(c)}$,edge label={node[midway,sloped,above]{w.p. $\varepsilon$}}
            ]
            ]
        \end{forest}
    \end{center}
    \caption{A process inducing $\bsigmaSTV^c$.}\label{fig:profile-stv}
    \end{figure}
    
    Observe that for any pair $c, c'$, the profiles $\bsigmaSTV^{c}$ and $\bsigmaSTV^{c'}$ are $(m-1)$-indistinguishable. Indeed, their generating processes are the same with probability $1-\varepsilon$ (when we sample uniformly from $R$), and with probability $\varepsilon$ they differ due to $\bsigma^{c}$ and $\bsigma^{c'}$ which are themselves $(m-1)$-indistinguishable.

     We next show that on each profile $\bsigmaSTV^c$, the unique STV winner is $c$. 
     %Without loss of generality, relabel the candidates in the cycle such that $c_m := c$ (so $c_1 = \text{next}(c)$, $c_2 = \text{next}(c_1)$, and so on). 
     Without loss of generality, consider the candidate $c = c_m$ (so $\text{next}(c) = c_1$, $\text{next}(c_1) = c_2$, and so on) as the argument holds symmetrically for any other $c$.
     We will show by strong induction that the $k$'th candidate to be eliminated is $c_k$. This implies that $c_m$ will be the final candidate remaining, and thus the STV winner.
     
     We begin with the base case, that $c_1$ is the first to be eliminated. Note that when sampling uniformly from $R$, by symmetry, each candidate has the same plurality score. On the other hand, in $\bsigma^{\text{next}(c)}$, candidate $\text{next}(c) = c_1$ is the unique plurality minimizer. Thus, in the mixed profile $\bsigmaSTV^c$, $c_1$ has the lowest plurality score and is eliminated first.
     
     Next, suppose candidates $c_1, \ldots, c_{k - 1}$ for $k \ge 2$ (and $k \le m - 1$) have been eliminated. We will show that the next to be eliminated is $c_k$. Let $C^k = C \setminus \set{c_1, \ldots, c_{k - 1}}$ be the set of uneliminated candidates. We first consider the proportion of first-place votes each candidate $c \in C^k$ gets in $\Unif(R)|_{C^k}$. Let $R = R^1 \sqcup \cdots \sqcup R^m$ be a partition of $R$ such that $R^i$ contains the rankings with $c_i$ ranked first. By symmetry, these are each the same size. Note that for $i \ge k$, $c_i \in C_k$, so rankings in $R^i$ will continue to rank $c_i$ first, each accounting for a $1/m$ proportion of the rankings. For $i \le k - 2$, since $\text{next}(c_i) \notin C^k$, by symmetry, the first place votes of $R^i$ will be distributed equally among all candidates in $C^k$. For $R^{k-1}$ however, every $\sigma \in R^{k - 1}$ ranks $c_k$ last, and, since $|C^k| \ge 2$, no votes will go to $c_k$; instead, they will be spread equally among $C^k \setminus \set{c_k}$. Hence, the plurality score of $c_k$ on $\text{Unif}(R)|_{C^k}$ will be $\frac{1}{m \cdot (|C^k| - 1)} \ge \frac{1}{m(m - 1)}$ smaller than all other candidates. By the choice of $\varepsilon$, $(1 - \varepsilon)\frac{1}{m(m - 1)} > \varepsilon$, so no matter how many first place votes $c_k$ gets on $\bsigma^{\text{next}(c)}$, $c_k$ has the smallest plurality score on $\bsigmaSTV^c|_{C^k}$. Therefore, it is the next to be eliminated. 

    Finally, we apply \Cref{lem2ToM} to this set of $m$ profiles to conclude that there exists a profile where no $(m-1)$-query algorithm can determine the STV winner with probability greater than $\tfrac{1}{m}$. 
   % We lastly apply a similar idea as in \Cref{lem2ToM}. Since we have $m$ profiles that are $m-1$ indistinguishable, any algorithm $A$ run of any of these instances will receive the exact same query response and thus must have identical output on all of them. Thus, there exists a candidate $a^*$ which it outputs with probability at most $\tfrac{1}{m}$. Hence $\bsigmaSTV^* = \bsigmaSTV^{a^*}$ satisfies the desired property.
\end{proof}

\section{Query Complexity}\label{secQC}

In the previous section, we characterized when it was information-theoretically possible to find winning candidates using a certain query size. We turn now to focusing on those cases when it \emph{is} possible and prove bounds on the \emph{query complexity}. In other words, when it is possible to determine the winner with limited-sized query size, how many such queries are needed?

For an integer $k \le m$, we use the notation $\binom{S}{k}$ to denote the set of all subsets of $S$ of size $k$.
Fix a scoring vector $\balpha$ and let $t^*$ be the minimal value such that $\balpha \in R_{m, t^*}$. Suppose we can make queries of size $t \ge t^*$ and wish to find a candidate maximizing $\score^\balpha_\bsigma$. As a benchmark, note that if we make enough queries to be able to deduce $\bsigma |_S$ for all possible $S \in \binom{C}{t^*}$, then it is information-theoretically possible to find this winning candidate. Indeed, one could simulate any $t^*$-query algorithm using this (say the one from \Cref{lemRmtClosedForm}) as they would have the responses for all $t$-sized query. Let $\text{cov}(m, t, t^*)$ be the minimum number of subsets of size $t$ needed to cover all subsets of size $t^*$ out of a set of size $m$. This value is referred to as a \emph{covering number}; computing such covering numbers and optimal subset structures that induce them is a canonical problem in combinatorics with a rich history (see, e.g., \citet{covering}). For our purposes, a reasonable (and nearly tight) lower bound on covering numbers is
$$\text{cov}(m, t, t^*) \ge \frac{\binom{m}{t^*}}{\binom{t}{t^*}}.$$
This follows from a simple argument: There are $\binom{m}{t^*}$ subsets of size $t^*$, and each subset of size $t$ can cover at most $\binom{t}{t^*}$ subsets them. When $t$ and $t^*$ are treated as constant, then this value is $\Omega(m^{t^*})$.

Our primary question is whether we can cleverly choose queries to use fewer than $\text{cov}(m, t, t^*)$ queries. 
%, as it may not be reasonable to assume that a sufficient number of voters arrive \sh{don't like this motivation since we are in this idealized query world}. 
As a motivating example, consider instead finding a \emph{Condorcet winner} under our model. A Condorcet winner on profile $\bsigma$ is a candidate $a$ that beats all others in a pairwise competition. More formally, for all $b \ne a$, $\Pr_{\sigma \sim \bsigma}[a \succ_{\sigma} b] > 1/2$. Note that Condorcet winners need not exist,\footnote{The classic example is when a third of the voters have each of the rankings $a \succ b \succ c$, $c \succ a \succ b$, and $b \succ c \succ a$.} however when they do, they are unique. From the definition, we can see that using queries of size $t^* = 2$ is sufficient to determine whether a Condorcet winner exists and, if so, determine this candidate as this only depends on pairwise margins. %Suppose we can make queries of size $t = 2$. 
One option is to make all possible queries of size $2$; this requires $\text{cov}(m, 2, 2) = \binom{m}{2} = \Theta(m^2)$ queries. However, as shown by \citet{Pro08b}, there is a more clever way, requiring only $O(m)$ queries to compute this winner.\footnote{The algorithm runs in two phases. First, run a knockout tournament among the candidates, where the candidate receiving more pairwise votes makes it on to the next round. If there is a Condorcet winner, then that candidate must be the winner of the knockout tournament. In the second phase, compare this winner to all other candidates they did not play. If they win all of these comparisons, they are the Condorcet winner, if not, there is no winner. This requires $2m - \floor{\log m} - 2$ queries. }

We now ask whether such improvements can be found for scoring rules. Unfortunately, we show that for both deterministic and randomized algorithms, they cannot.

\begin{theorem}\label{thmCovering}
Fix $\balpha \in \mathbb{R}^m$ and let $t^*$ be the minimal value such that $\balpha \in R_{m, t^*}$. For $t^* \ge 2$, any deterministic $t$-query algorithm that always outputs a candidate maximizing $\score^{\balpha}_{\bsigma}$ must make at least $\text{cov}(m, t, t^*)$ queries in the worst case\footnote{It is only when $\balpha$ is a constant vector does $t^* = 1$, a degenerate case that we ignore since any candidate can be considered the winner.}. Further, any randomized algorithm making at most $\delta \frac{\binom{m}{t^*}}{\binom{t}{t^*}}$ queries outputs an $\balpha$-winner with probability at most $\min(\delta + \frac{1}{m}, \delta + (1 - \delta) \frac{1}{t^*})$ in the worst case.
\end{theorem}
\begin{proof}
Fix $\balpha$ and $t^*$. Let $C_1$ be a set of candidates of size $t^*$ with two distinguished candidates $a, b \in C_1$, and let $C_2 = \overline{C_1}$ be the remaining candidates. We construct $\bsigma^{1}, \ldots, \bsigma^{m - t^*}$ just as in \Cref{fig:sigmai} from \Cref{subsec:characterization}. Recall that each profile $\bsigma^{i}$ has a profile $\bsigma \in \Pi(C_1)$ satisfying the conditions of \Cref{lem:construction} with $a$ and $b$ ``contained'' in it. In addition, it has $i - 1$ of the $C_2$ candidates ranked in a fixed order before $\bsigma$, and the rest are in a fixed order after. After describing these profiles in \Cref{subsec:characterization}, we observed that each $\bsigma^{i}$ and $(\bsigma^{i})^\atob$ are $(t^*-1)$-indistinguishable. 
However, we claim that an even stronger property is true: For any $t$ sized query $Q$ that does not contain $C_1$, $\bsigma^{i}$ and $(\bsigma^{i})^\atob$ are indistinguishable (note $t$ may be much larger than $t^*$). More formally, for all $Q$ such that $C_1 \not\subseteq Q$, $\bsigma^{i} |_Q = (\bsigma^{i})^\atob |_Q$. 
Indeed, the candidates of $C_2 \cap Q$ are always in the same order, either before or after the candidates of $C_1$. 
The candidates in $C_1$ will follow the distribution according to $\bsigma |_{Q \cap C_1}$ and $\bsigma^\atob | _{Q \cap C_1}$. 
By \Cref{lem:construction}, since $Q \cap C_1 \subsetneq C_1$, these are identical.

By definition of $t^*$, $\balpha \notin R_{m, t^* - 1}$. Hence, \Cref{lemCharacterizationForwardDirection} implies that for one of these profiles, $\score^{\balpha}_{\bsigma^{i}}(a) \ne \score^{\balpha}_{\bsigma^{i}}(b)$. Fix such an $i$, and without loss of generality, assume $\score^{\balpha}_{\bsigma^{i}}(a) > \score^{\balpha}_{\bsigma^{i}}(b)$. Consider the profile $\bsigma^*$ induced by sampling a permutation $\pi$ uniformly at random, and, if $\pi(c) = c$ for all $c \in C_1$, sample $\sigma \sim \pi \circ (\bsigma^{ i})^\atob$,  otherwise, sample $\sigma \sim \pi \circ \bsigma^{i}$. This is shown in \Cref{fig:profile-query}.

\begin{figure}
    \centering
    \begin{minipage}[c]{\textwidth}
         \begin{center}
             \begin{forest}
                for tree={
                grow=east, % tree direction
                parent anchor=east,
                child anchor=west, % edge anchors
                rounded corners, draw,
                align=center,
                edge={->},
                s sep = 8ex
                }
                [{Choose permutation $\pi$\\uniformly at random},l sep = 30 ex,
                [
                    Sample $\sigma \sim \pi \circ \bsigma^i$,edge label={node[midway,sloped,above]{$\exists c \in C_1, \pi(c) \ne c$}}
                ]
                [
                    Sample $\sigma \sim \pi \circ (\bsigma^i)^\atob$,edge label={node[midway,sloped,above]{$\forall c \in C_1, \pi(c) = c$}}
                ]
                ]
            \end{forest}
    \end{center}
    \caption{A process inducing $\bsigma^*$.}\label{fig:profile-query}
    \vspace{2em}
    \end{minipage}
    
    \begin{minipage}[c]{\textwidth}
        \begin{center}
            \begin{forest}
                for tree={
                grow=east, % tree direction
                parent anchor=east,
                child anchor=west, % edge anchors
                rounded corners, draw,
                align=center,
                edge={->},
                s sep = 8ex
                }
                [{Choose permutation $\pi$\\uniformly at random},l sep = 10ex,anchor=east
                [
                    {Sample $\sigma \sim \pi \circ \bsigma^i$},anchor=west
                ]
                ]
            \end{forest}
        \end{center}
        \caption{A process inducing $\bsigma^\unif$.}\label{fig:profile-query-unif}
    \end{minipage}
\end{figure}
    
Next, we will compare $\bsigma^*$ to another profile, $\bsigma^{\unif}$, defined in \Cref{fig:profile-query-unif}, where we sample from $\pi \circ \bsigma^i$ regardless of $i$. Note that $\bsigma^{\unif}$ is the uniform distribution over all rankings which can be equivalently achieved by first sampling $\sigma \sim \bsigma^i$, and then outputting $\pi \circ \sigma$ for a $\pi$ that is uniformly selected.\footnote{Note that although the generating process for $\bsigma^\unif$ here is different than the one used in \Cref{lemma:m_indist} and \Cref{fig:profile-unif}, the resulting distribution is still the same.} We will show two things (i) $\bsigma^*|_Q = \bsigma^{\unif}$ unless $C_1 \subseteq Q$, and (ii) $b$ is the unique $\balpha$-winner on $\bsigma^*$.  Note that (i) implies that on any query $Q$ with $C_1 \not\subseteq Q$, $\bsigma^*|_Q$ is the uniform distribution over rankings in $\mathcal{L}(Q)$.

For (i), fix a query $Q$ with $C_1 \not\subseteq Q$. Since the $\pi(c) \ne c$ for some $c \in C_1$ branch of $\bsigma^*$ is identical to $\bsigma^\unif$, it suffices to focus on $\pi$ such that $\pi(c) = c$ for all $c \in C_1$. When this holds, $\pi$ can only reorder the candidates of $C_2$, leaving candidates in $C_1$ unchanged. Therefore, $(\pi \circ (\bsigma^i)^\atob)|_Q = (\pi \circ \bsigma^i)|_Q$. Hence, sampling from $\bsigma^*|_Q$ is equivalent to sampling $\pi$ uniformly at random and sampling from $(\pi \circ \bsigma^i)|_Q$, equivalent to sampling from $\bsigma^\unif |_Q$.

For (ii), we will show that $\score^\balpha_{\bsigma^*}(b) > \score^\balpha_{\bsigma^\unif}(b)$ while $\score^\balpha_{\bsigma^*}(c) \le \score^\balpha_{\bsigma^\unif}(c)$ for all $ c\ne b$. By symmetry, the scores of all candidates in $\bsigma^\unif$ are the same, hence, this shows that $b$ is the unique winner. To that end, note that the only time in sampling $\bsigma^*$ and $\bsigma^\unif$ that the scores will differ, is if we take the top branch, in which case $\pi(c) = c$ for all $c \in C_1$. By assumption, $\pi(a) = a$ and $\pi(b) = b$ as  $a, b \in C_1$, so we are simply swapping $a$ and $b$. Since $\score^{\balpha}_{\bsigma^i}(a) > \score^{\balpha}_{\bsigma^i}(b)$, this strictly increases the score of $b$ on average, and strictly decreases the score of $a$. Hence, $\score^\balpha_{\bsigma^*}(b) > \score^\balpha_{\bsigma^\unif}(b)$, $\score^\balpha_{\bsigma^*}(a) < \score^\balpha_{\bsigma^\unif}(a)$, and $\score^\balpha_{\bsigma^*}(c) = \score^\balpha_{\bsigma^\unif}(c)$ for all $c \ne a, b$, as needed.

Fix a deterministic algorithm $t$-query algorithm $\mathcal{A}$ that outputs a candidate after making strictly fewer than $\text{cov}(m, t, t^*)$ queries. We will show that it cannot always output an $\balpha$-winner. Consider a run of the algorithm where on every query $Q$, it receives in response the uniform distribution over $\mathcal{L}(Q)$.
Suppose on this run, it outputs candidate $c$. Now, by the definition of the covering number, there must be a set $C'$ with $|C'| = t^*$ such that $C'$ was not contained in any query made by the algorithm. Since $t^* \ge 2$, $C' \setminus \set{c}$ is not empty. Let $c^* \in C' \setminus \set{c}$.
Let $\pi$ be a permutation such that $\pi(b) = c^*$ and $\pi$ maps $C_1 \setminus \set{b}$ to $C' \setminus \set{c^*}$.
Consider the running $\mathcal{A}$ on $\pi \circ \bsigma^*$. Note that on every query $Q$ not containing $C_1$, the response will be indistinguishable from $\bsigma^\unif$, and hence, the uniform distribution over $Q$. Therefore, by above, $\mathcal{A}$ will return candidate $c$ on this instance. However, by construction, $c^*$ is the unique $\balpha$-winner, and $c^* \ne c$, a contradiction.

Next, we will show that a randomized algorithm making at most $\delta \frac{\binom{m}{t^*}}{\binom{t}{t^*}}$ queries will output an $\balpha$-winner with probability at most $\min(\delta + \frac{1}{m}, \delta + (1 - \delta) \frac{1}{t^*})$. We will make use of Yao's Minimax Principle~\citep{Yao77}. More specifically, will show that there is a distribution over profiles such that no deterministic algorithm can be correct with larger probability. This implies that no randomized algorithm can achieve a larger probability on a worst-case profile.

The distribution over instances we will choose is simply uniform over $\pi \circ \bsigma^*$ for all permutations $\pi$. Fix an arbitrary deterministic algorithm $\mathcal{A}$ that always outputs a candidate after at most $\delta \binom{m}{t^*}/\binom{t}{t^*}$ queries. Consider a run of this algorithm where the response to every query $Q$ is the uniform distribution over $\mathcal{L}(Q)$, and suppose on this run, the output candidate is $c$. Let $\mathcal{Q}$ be the set of queries asked on this run. We have that $|\mathcal{Q}| \le \delta \binom{m}{t^*}/\binom{t}{t^*}$ by assumption. Observe that since each query of size $t$ can cover $\binom{t}{t^*}$ sets of size $t^*$, at most a $\delta$-fraction of the $\binom{m}{t^*}$ $t^*$-sets are covered by $\mathcal{Q}$.

We claim that the algorithm must be incorrect for all $\pi \circ \bsigma^*$ such that both (i) $\pi(C_1) \not\subseteq Q$ for all $Q \in \mathcal{Q}$ and (ii) $\pi(b) \ne c$. Indeed, (i) ensures that the run of the algorithm will always lead to uniform responses, meaning $\mathcal{A}$ \emph{must} output $c$, and (ii) ensures this is the incorrect choice. More formally, let $$\mathcal{E}_1 = \set{\pi \mid \pi(C_1) \not\subseteq Q \text{ for all } Q \in \mathcal{Q}}$$ be the event that the first property holds, and $$\mathcal{E}_2 = \set{\pi \mid \pi(b) \ne c}$$ be the event that the second does. The probability of success is at most $1 - \Pr[\mathcal{E}_1 \cap \mathcal{E}_2]$. We will upper bound this in two ways. First,
% \begin{align*}
%   1 - \Pr[\mathcal{E}_1 \cap \mathcal{E}_2] &= \Pr[\overline{\mathcal{E}_1} \cup \overline{\mathcal{E}_2}]\\
%   &\le \Pr[\overline{\mathcal{E}_1}] + \Pr[\overline{\mathcal{E}_2}]\\
%   &\le \delta + \frac{1}{m},
% \end{align*}
\begin{equation*}
    1 - \Pr[\mathcal{E}_1 \cap \mathcal{E}_2] = \Pr[\overline{\mathcal{E}_1} \cup \overline{\mathcal{E}_2}]
    \le \Pr[\overline{\mathcal{E}_1}] + \Pr[\overline{\mathcal{E}_2}]
    \le \delta + \frac{1}{m},
\end{equation*}
where the first inequality holds by the union bound, $\Pr[\overline{\mathcal{E}_1}] \le \delta$ because at most a $\delta$-fraction of all $t^*$-subsets are covered, and  $\Pr[\overline{\mathcal{E}_2}] \le \frac{1}{m}$ by symmetry. Second,
\begin{align*}
1 - \Pr[\mathcal{E}_1 \cap \mathcal{E}_2]
    &= 1 - \Pr[\mathcal{E}_2 \mid \mathcal{E}_1] \cdot \Pr[\mathcal{E}_1]\\
    &\le 1 - \left(1 - \frac{1}{t^*}\right) \cdot (1 - \delta)\\
    &= 1 - (1 - \delta) + (1 - \delta) \cdot \frac{1}{t^*}\\
    &= \delta + (1 - \delta) \cdot \frac{1}{t^*}.
\end{align*}
Again, $\Pr[\mathcal{E}_1] \ge 1 - \delta$ holds because at most a $\delta$-fraction are covered. The other term $\Pr[\mathcal{E}_2 \mid \mathcal{E}_1] \cdot \Pr[\mathcal{E}_1]$ holds because conditioned on $\pi(C_1) = C'$ for \emph{any} uncovered $C'$, the probability that $\pi(b) = c$ is at most $1 - \frac{1}{t^*}$. Indeed, this holds exactly if $c \in C'$, and is $0$ otherwise. Together, these two bounds imply that the probaility of success is at most $\min(\delta + \frac{1}{m}, \delta + (1 - \delta) \frac{1}{t^*})$.
\end{proof}

\begin{figure}
  \begin{minipage}[c]{0.615\textwidth}
    \includegraphics[width=\textwidth]{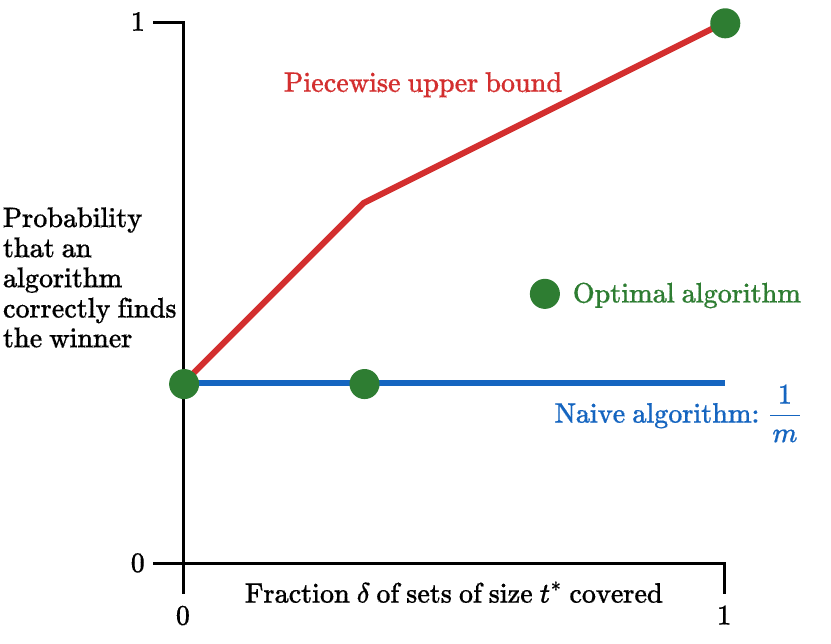}
  \end{minipage}\hfill
  \begin{minipage}[c]{0.355\textwidth}
    \caption{
       Success probabilities for randomized algorithms computing the Borda winner on $m = 3$ candidates making queries of size $t = t^* = 2$. The red piecewise-linear curve, which has similar shapes for larger parameters and other computable scoring rules, represents the upper bound on the success probability of any algorithm from Theorem~\ref{thmCovering}. The blue curve is the best-known general lower bound obtained by the algorithm that makes no queries and simply guesses a candidate at random. The green points are the true success probabilities of the optimal algorithm given by Theorem~\ref{thmQC32}.
    }
    \label{figQueryComplexityBounds}
  \end{minipage}
\end{figure}

% \begin{figure}
% \includegraphics[width=.55\textwidth]{QueryComplexityBounds32-3.pdf}
 
%     \caption{
%        Success probabilities for randomized algorithms computing the Borda winner on $m = 3$ candidates making queries of size $t = t^* = 2$. The red piecewise-linear curve, which has similar shapes for larger parameters and other computable scoring rules, represents the upper bound on the success probability of any algorithm from Theorem~\ref{thmCovering}. The blue curve is the best-known general lower bound obtained by the algorithm that makes no queries and simply guesses a candidate at random. The green points are the true success probabilities of the optimal algorithm given by Theorem~\ref{thmQC32}.
%     }
%     \label{figQueryComplexityBounds}
% \end{figure}

%\ipnc{.67}{QueryComplexityBounds32-2}{\label{figQueryComplexityBounds} Success probabilities for randomized algorithms computing the Borda winner on $m = 3$ candidates making queries of size $t = t^* = 2$. The piecewise-linear curves in the picture have similar shapes for larger parameters and other computable scoring rules. The upper red curve is the upper bound on the success probability of any algorithm from Theorem \ref{thmCovering}. The blue curve is the best known general lower bound, obtained by the algorithm that makes no queries and simply guesses a candidate at random. The green points are the true success probabilities of the optimal algorithm given by Theorem \ref{thmQC32}.}

Theorem~\ref{thmCovering} completely settles the query complexity of deterministic algorithms for computing all positional scoring rules. However, for randomized algorithms, the story is not quite complete. On the one hand, if $t^*$ and $t$ are constants, and an algorithm would like to be correct with probability $\frac1m + c$ for a constant $c$, then $\Omega(m^{t^*})$ queries are needed (hiding constants depending on $t^*, t$, and $c$), and $O(m^{t^*})$ clearly suffice, as regardless of $t$, $\binom{m}{t^*}$ are certainly enough to cover all sets. On the other hand, if an algorithm can make $\delta \frac{\binom{m}{t^*}}{\binom{t}{t^*}}$ queries for a fixed $\delta$, we do not know the exact probability with which it can be correct. Figure~\ref{figQueryComplexityBounds} depicts the gap between the upper bound and the best general lower bound as functions of the parameter $\delta$. This lower bound is essentially the naive algorithm achieving $\frac1m$ by simply picking a candidate at random. However, as the following result shows, even for the simplest nontrivial case of $m = 3$ and $t = t^* = 2$, the true query complexity of computing the (essentially unique) scoring rule in $R_{3, 2}$ is strictly between our general bounds when $\delta = \frac23$.

% We have space and this is easier to read when the theorem statement isn't broken up as it currently does without newpage

\begin{theorem}\label{thmQC32}
    With $m = 3$ candidates, the optimal randomized algorithm making queries of size $t = 2$ to compute the Borda count winner succeeds with
    \begin{enumerate}
        \item\label{itmQC32-1} worst-case success probability $\frac13$ when allowed to make exactly one query, and
        \item\label{itmQC32-2} worst-case success probability $\frac12$ when allowed to make exactly two queries.
    \end{enumerate}
\end{theorem}

Before giving the proof, we note that the measure of worst-case optimality here is a bit finicky. On the negative side, we show that for all algorithms and any $\varepsilon > 0$, there are instances where they do not succeed with probability more than $\frac13 + \varepsilon$ for (\ref{itmQC32-1}), and $\frac12 + \varepsilon$ for (\ref{itmQC32-2}). At least for (\ref{itmQC32-1}), this relaxation is necessary. Consider the algorithm that makes a single query to a uniformly random pair of candidates and selects between them with the probabilities given by the query response (i.e., if the algorithm learns that $\Pr_{\sigma \sim \bsigma}[a \succ_\sigma b] = p$, it picks $a$ with probability $p$ and $b$ with probability $1-p$). It can be shown that this process selects each candidate with probability proportional to its Borda score.\footnote{The probability it picks a candidate $c$ is equal to $\frac{1}{\binom{m}{2}}\sum_{c' \ne c} \Pr_{\sigma \sim \bsigma}[c \succ c']$. It is well known that the Borda score of $c$ is equal to $\sum_{c' \ne c} \Pr_{\sigma \sim \bsigma}[c \succ c']$~\citep{Handbook}.} Given any fixed profile, unless all three candidates have the same score, a maximal one will be selected with probability strictly greater than $\frac13$ (and if they are all the same, then all are Borda winners, and hence the algorithm succeeds with probability $1$). However, there are instances where the best Borda score is arbitrarily close to the others, resulting in a success probability no constant greater than $\frac13$. Thus, by saying that ``the optimal'' randomized algorithm that makes a single query achieves a worst-case success probability of $\frac13$, we really mean that it is not possible to surpass $\frac13$ by any constant. In the proof of Theorem~\ref{thmQC32}, we must construct a family of increasingly more difficult instances that bring the success probabilities closer to $\frac13$ and $\frac12$. This becomes quite complicated, involving a construction based on Fibonacci numbers to ensure query responses do not leak cardinal information about the relative strengths of candidates.

\paragraph{Proof of \Cref{thmQC32}:} Fix a set of candidates $C = \set{a, b, c}$.
Throughout the proof, we use the scoring vector $(1, 0, -1)$ to compute Borda scores (which is equivalent to the more traditional choice of $(2, 1, 0)$ by translation). When writing scores, we drop the $(1, 0, -1)$ superscript in the score notation, using $\score_\bsigma(c')$ to refer to $\score^{(1, 0, -1)}_\bsigma(c')$. We also use the convention that when $\bsigma$ is queried on $\set{a, b}$, the algorithm learns $\Pr_{\sigma \sim \bsigma}[a \succ_\sigma b]$, on $\set{b, c}$, the algorithm learns $\Pr_{\sigma \sim \bsigma}[b \succ_\sigma c]$, and on $\set{a, c}$, the algorithm learns $\Pr_{\sigma \sim \bsigma}[c \succ_\sigma a]$. These single numbers completely parameterize the distribution $\bsigma|_Q$ for each $Q$ of size 2. One can check that the following equalities hold for scores
\begin{align*}
    \score_\bsigma(a) &= \Pr_{\sigma \sim \bsigma}[a \succ_\sigma b] - \Pr_{\sigma \sim \bsigma}[b \succ_\sigma c],\\
    \score_\bsigma(b) &= \Pr_{\sigma \sim \bsigma}[b \succ_\sigma c] - \Pr_{\sigma \sim \bsigma}[c \succ_\sigma a],\\
    \score_\bsigma(c) &= \Pr_{\sigma \sim \bsigma}[c \succ_\sigma a] - \Pr_{\sigma \sim \bsigma}[a \succ_\sigma b].
\end{align*}

We begin with the lower bounds on the probabilities. First, note that it is always possible to succeed with probability $\frac13$ by just picking a random one of the three candidates. This establishes the lower bound on (\ref{itmQC32-1}).
    
For (\ref{itmQC32-2}), consider the following algorithm. We pick a random candidate $c'$ and query both sets of size 2 containing that candidate. From this information, we are able to learn the Borda score of candidate $c'$. If the score is positive, we return $c'$. Otherwise, we randomly return one of the other two candidates. Observe that, with our choice of scoring vector $(1, 0, -1)$, the sum of all three Borda scores must be zero, so at most two are strictly positive. If none of them are positive, then they must all be zero, in which case every candidate is a Borda winner, and the algorithm succeeds with probability 1. If one Borda score is positive, then if that candidate is chosen as $c$, the algorithm succeeds with probability 1, and otherwise the algorithm succeeds with probability $\frac12$. Since the former case happens with probability $\frac13$, the total expected success probability is $\frac23$. Finally, if two Borda scores are positive, then if the true Borda winner is chosen as $c'$, the algorithm succeeds with probability 1; if the other candidate with positive Borda score is chosen as $c'$, the algorithm incorrectly returns it, succeeding with probability 0; and if the candidate with negative Borda score is chosen as $c$, the algorithm succeeds with probability $\frac12$. In total, the success probability is $\frac13 \cdot 1 + \frac13 \cdot 0 + \frac13 \cdot \frac12 = \frac12$. Thus, the worst-case success probability is $\frac12$.

For the upper bounds, we again use Yao's Minimax principle. That is, we will show that there are distributions over instances where no deterministic algorithm can output a Borda winner with probability more than $\frac13 + \varepsilon$ for any $\varepsilon > 0$. This implies that no randomized algorithm can do so on every instance.

Define a family of distributions over profile $D_1, D_2, D_3, \dots$ as follows. To generate $D_n$, we first sample three values, $p_1, p_2, p_3 \in [\frac13, \frac23]$ and return an arbitrary profile $\bsigma$ where
\begin{align*}
    p_1 &:= \Pr_{\sigma \sim \bsigma} [a \succ_\sigma b],\\
    p_2 &:= \Pr_{\sigma \sim \bsigma} [b \succ_\sigma c],\\
    p_3 &:= \Pr_{\sigma \sim \bsigma} [c \succ_\sigma a].
\end{align*}
Before showing how to sample $p_1, p_2, p_3$, we first show that such a profile $\bsigma$ satisfying the pairwise margins always exists. We claim that the following preference profile suffices.
\begin{table}[H]
    \centering
    \begin{tabular}{r|l}
        Ranking & Probability \\\hline
        $a \succ b \succ c$ & $p_2 - \frac13$\\
        $a \succ c \succ b$ & $\frac23 - p_2$\\
        $b \succ c \succ a$ & $p_3 - \frac13$\\
        $b \succ a \succ c$ & $ \frac23 - p_3$\\
        $c \succ a \succ b$ & $ p_1 - \frac13$\\
        $c \succ b \succ a$ & $ \frac23 - p_1$
    \end{tabular}
\end{table}
The reader may verify that:
\begin{itemize}
    \item All probabilities lie in $[0, 1]$ for $p_1, p_2, p_3 \in [\frac13, \frac23]$.
    \item The sum of all six probabilities is 1.
    \item The pairwise ranking probabilities are indeed given by $p_1$, $p_2$, and $p_3$.
\end{itemize}

As shown above, the $p_i$ values contain all relevant information for computing the Borda scores of each candidate: $\score_\bsigma(a) = p_1 - p_3$, $\score_\bsigma(b) = p_2 - p_1$, and $\score_\bsigma(c) = p_3 - p_2$, along with the responses for all the queries. Hence, the exact construction of $\bsigma$ will not be important for the remainder of the proof.

Let $F_1, F_2, F_3, \dots$ be the Fibonacci sequence shifted to the left by one, beginning with $F_1 := 1$, $F_2 := 2$, $F_3 := 3$, $F_4 := 5$, and so on. For any positive integer $n$, we generate $p_1$, $p_2$, and $p_3$ for $D_n$ using the following process. First, sample $i$ uniformly from $\{1, 2, \dots, n\}$, sample $s$ uniformly from $\{0, 1, 2, \dots, nF_{n + 2}\}$ (all integers from $0$ to $nF_{n + 2}$), and sample $r$ uniformly from $\{1, 2, 3, 4, 5, 6\}$. Then output the profile $\bsigma(i, s, r)$ defined by the table below, where $p_1$, $p_2$, and $p_3$ are defined from the auxiliary values $\hat{p}_1$, $\hat{p}_2$ and $\hat{p}_3$ by the correspondence
$$p_j := \frac13 + \frac{1}{3}\cdot\frac{1}{((n + 1)F_{n + 2})} \cdot \hat{p}_j.$$
\begin{table}[H]
    \centering
    \begin{tabular}{r|c|l}
        Profile & True Borda winner & Scaled probabilities $(\hat{p}_1, \hat{p}_2, \hat{p}_3)$ \\\hline
        $\bsigma(i, s, 1)$ & $a$ & $\left(s + F_{i + 2}, s, s + F_{i}\right)$\\
        $\bsigma(i, s, 2)$ & $c$ & $\left(s + F_{i + 2}, s, s + F_{i + 1}\right)$\\
        $\bsigma(i, s, 3)$ & $b$ & $\left(s + F_{i}, s + F_{i + 2}, s\right)$\\
        $\bsigma(i, s, 4)$ & $a$ & $\left(s + F_{i + 1}, s + F_{i + 2}, s\right)$\\
        $\bsigma(i, s, 5)$ & $c$ & $\left(s, s + F_{i}, s + F_{i + 2}\right)$\\
        $\bsigma(i, s, 6)$ & $b$ & $\left(s, s + F_{i + 1}, s + F_{i + 2}\right)$
    \end{tabular}
\end{table}
Note that each $p_j \in [\frac13, \frac23]$ if and only if $\hat{p}_j \in [0, (n + 1)F_{n + 2}]$, so any pair $(i, s)$ is valid. We leave the reader to verify that the Borda winners are as stated in the table. We explicitly compute the winner in the final profile as an example:
\begin{align*}
    \score_{\bsigma(s, i, 6)}(a) &= p_1 - p_3 = \frac13 + \frac{1}{3}\cdot\frac{1}{((n + 1)F_{n + 2})} \cdot (\hat{p}_1 - \hat{p}_3) \\&= \frac13 + \frac{1}{3}\cdot\frac{1}{((n + 1)F_{n + 2})} \cdot ((s) - (s + F_{i + 2})) = \frac13 + \frac{1}{3}\cdot\frac{1}{((n + 1)F_{n + 2})} \cdot (-F_{i + 2})\\
    \score_{\bsigma(s, i, 6)}(b) &= p_2 - p_1 = \frac13 + \frac{1}{3}\cdot\frac{1}{((n + 1)F_{n + 2})} \cdot (\hat{p}_2 - \hat{p}_1) \\&= \frac13 + \frac{1}{3}\cdot\frac{1}{((n + 1)F_{n + 2})} \cdot ((s + F_{i + 1}) - (s)) = \frac13 + \frac{1}{3}\cdot\frac{1}{((n + 1)F_{n + 2})} \cdot (F_{i + 1})\\
    \score_{\bsigma(s, i, 6)}(c) &= p_3 - p_2 = \frac13 + \frac{1}{3}\cdot\frac{1}{((n + 1)F_{n + 2})} \cdot (\hat{p}_3 - \hat{p}_2) \\&= \frac13 + \frac{1}{3}\cdot\frac{1}{((n + 1)F_{n + 2})} \cdot ((s + F_{i + 2}) - (s + F_{i + 1})) = \frac13 + \frac{1}{3}\cdot\frac{1}{((n + 1)F_{n + 2})} \cdot (F_i)
\end{align*}
Thus, the winner is $b$, since $F_{i + 1}$ is the largest out of $\{-F_{i + 2}, F_{i + 1}, F_{i}\}$.

Fix a deterministic algorithm $\mathcal{A}$ making at most one query of size $2$. Observe that $D_n$ is completely symmetric with respect to $p_1$, $p_2$, and $p_3$ (in the sense that permuting $p_1 \mapsto p_2$, $p_2 \mapsto p_3$, and $p_3 \mapsto p_1$ gives the same distribution on preference profiles). Thus, we may assume without loss of generality that an algorithm queries $\set{a, b}$ and learns $\hat{p}_1 \in \{0, 1, 2, \dots, (n + 1)F_{n + 2}\}$. Let $\mathcal{E}$ be the event that $i \in \{3, 4, 5, \dots, n - 2\}$ and $s \in [F_{n + 2}, (n - 2)F_{n + 2}]$. By the union bound, the probability that $\mathcal{E}$ does not occur is at most
\begin{align*}
    \Pr[\overline{\mathcal{E}}] &\leq \Pr[i \in \{1, 2, n - 1, n\}] + \Pr[s \in [0, F_{n + 2}] \cup [(n - 2)F_{n + 2}, nF_{n + 2}]]\\
    &\leq \frac{4}{n} + \frac{F_{n + 2} + 2F_{n + 2}}{nF_{n + 2}}\\
    &= \frac7n.
\end{align*}
When $\mathcal{E}$ occurs, we will have $\hat{p}_j \in [F_{n + 2}, (n - 1)F_{n + 2}]$ for each $j$, so in particular this holds for $\hat{p}_1$. Suppose that the algorithm additionally learns $i$, which only makes it stronger. Then there are exactly six possible choices of the parameters $s$ and $c$ that could have led to the specific realization $\hat{p}_1$:
\begin{itemize}
    \item $r = 1$ and $s = \hat{p}_1 - F_{i + 2}$ $\implies a$ is the winner.
    \item $r = 2$ and $s = \hat{p}_1 - F_{i + 2}$ $\implies c$ is the winner.
    \item $r = 3$ and $s = \hat{p}_1 - F_{i}$ $\implies b$ is the winner.
    \item $r = 4$ and $s = \hat{p}_1 - F_{i + 1}$ $\implies a$ is the winner.
    \item $r = 5$ and $s = \hat{p}_1$ $\implies c$ is the winner.
    \item $r = 6$ and $s = \hat{p}_1$ $\implies b$ is the winner.
\end{itemize}
Since each possibility is equally likely, no matter which candidate the algorithm picks it succeeds with probability $\frac13$. Thus, overall,

\begin{align*}
    \Pr[\txt{success}] &= \Pr[\mathcal{E}]\Pr[\txt{success} \suchthat \mathcal{E}] + \Pr[\overline{\mathcal{E}}]\Pr[\txt{success} \suchthat \overline{\mathcal{E}}]\\
    &\leq \Pr[\txt{success} \suchthat \mathcal{E}] + \Pr[\overline{\mathcal{E}}]\\
    &\leq \frac13 + \frac7n.
\end{align*}
By picking $n$ sufficiently large, we see that no algorithm that makes a single query can achieve a worst-case success probability of $\frac13 + \varepsilon$ for any $\varepsilon > 0$. 

For an algorithm that makes two queries, we similarly assume without loss of generality that the queries yielded the values of $\hat{p}_1$ and $\hat{p}_2$ in some order. Regardless of whether the second query was made adaptively or nonadaptively, we will argue that, from these responses the algorithm cannot determine the winner with probability greater than $\frac12 + \frac7n$. As before, it suffices to show that the algorithm succeeds with probability at most $\frac12$ when $\hat{p}_1, \hat{p}_2 \in [F_{n + 2}, (n - 1)F_{n + 2}]$. Here there are two cases to consider, depending on which observed value is larger.

First suppose $\hat{p}_1 > \hat{p}_2$, and let $j \in \{3, 4, 5, \dots, n\}$ be such that $\hat{p}_1 - \hat{p}_2 = F_j$. Then there are exactly two possible choices of the parameters $s$, $i$, and $c$ that could have led to the specific realizations $\hat{p}_1$ and $\hat{p}_{2}$:
\begin{itemize}
    \item $r = 1$, $i = j - 2$, and $s = \hat{p}_1 - F_{i + 2}$ $\implies a$ is the winner.
    \item $r = 2$, $i = j - 2$, and $s = \hat{p}_1 - F_{i + 2}$ $\implies c$ is the winner.
\end{itemize}
Thus, $a$ and $c$ are equally likely to be the winner, so no matter which candidate the algorithm returns, it will be correct with probability at most $\frac12$.

Now suppose $\hat{p}_1 < \hat{p}_2$, and let $j \in \{3, 4, 5, \dots, n\}$ be such that $\hat{p}_2 - \hat{p}_1 = F_j$. Then there are exactly four possibilities for $s$, $i$, and $c$:
\begin{itemize}
    \item $r = 3$, $i = j - 1$, and $s = \hat{p}_1 - F_{i + i}$ $\implies b$ is the winner. Note that this is the first place we make use of the Fibonacci recurrence: $\hat{p}_2 - \hat{p}_1 = (s + F_{i + 2}) - (s + F_i) = F_{i + 1} = F_j$.
    \item $r = 4$, $i = j$, and $s = \hat{p}_1 - F_{i + 1}$ $\implies a$ is the winner (again using the Fibonacci recurrence).
    \item $r = 5$, $i = j$, and $s = \hat{p}_1$ $\implies c$ is the winner.
    \item $r = 6$, $i = j - 1$, and $s = \hat{p}_1$ $\implies b$ is the winner.
\end{itemize}
Clearly, $b$ is the best guess here, but it is still only correct with probability $\frac12$.

Thus, picking $n$ sufficiently large as before, we conclude that no algorithm making only two queries can achieve a worst-case success probability of $\frac12 + \varepsilon$ for any $\varepsilon > 0$. \qed 

\section{Discussion}

Voting rules are increasingly applied to aggregate preferences across a large range of candidates, from primary elections to online opinions. This can, however, be at odds with the cognitive and implementation challenges that exist when requiring individuals to specify preferences across a large selection. Naturally, in such scenarios, voters end up specifying preferences over a limited set, whether by explicit design or implicitly by submitting incomplete votes. Our work studies the implications of this phenomenon on the computability of voting rules.

For the large class positional scoring rules, we provide an exact information-theoretic characterization of what can and cannot be correctly computed under this model. Specifically, a decrease in query size equivalently diminishes the dimension of the computable scoring vector space. We explicitly characterize these spaces, finding that, while the Borda count is included for $t \ge 2$, Plurality is not for \emph{any} limited-sized query. We also extend this strong impossibility to STV. From a practical perspective, these results demonstrate the pitfalls of common rules like Plurality and STV within the setting incomplete votes and point to the space of alternative rules. Future work could go beyond voting and further investigate this question of computability with limited-sized queries for other social choice rules or other more general functions, such as committee selection with rankings.

For rules computable with at least $t^*$-sized queries, we also give bounds on the query complexity of any deterministic or randomized algorithm making $t \geq t^*$ sized queries. While we show that deterministic algorithms must cover the space of all $t^*$ sized queries in the worst-case, thus giving a tight bound, the picture for randomized algorithms is far less clear. We give an upper bound on the success probability when using a given number of queries, yet no known general-purpose algorithm achieves anything close to it. In \Cref{thmQC32}, we close this gap for a special case of the Borda rule by constructing surprisingly intricate hard instances. Closing the general query complexity gap in our randomized setting is an intriguing open problem whose technical depth is illustrated by this result.

\bibliographystyle{ACM-Reference-Format}
\bibliography{abb,references}

\end{document}

%% file: JTex6.tex
\usepackage{amsmath}
\usepackage{graphicx}
\usepackage{float}
\usepackage{caption}
\usepackage{amsthm}
\usepackage{color}
\definecolor{green}{rgb}{0,0.5977,0}

\newcommand{\rr}{\mathbb R}

\newcommand{\abs}[1]{\left|{#1}\right|}

\newcommand{\suchthat}{\ | \ }

\newcommand{\tth}{^\text{th}}

\DeclareMathOperator*{\argmin}{argmin}

\newcommand{\txt}[1]{\text{#1}}

\newcommand{\stext}[1]{\ \ \ \ \ \text{(#1)}}

\makeatletter 
\g@addto@macro{\@algocf@init}{\SetKwInOut{Parameter}{Parameters}} 
\makeatother